\documentclass[12pt]{article}

\usepackage[letterpaper,margin=1in]{geometry}

\usepackage{amsmath,amssymb,amsthm}
\usepackage{mathtools}
\usepackage{bm}
\usepackage{bbm}
\usepackage{commath}

\numberwithin{equation}{section}
\allowdisplaybreaks

\usepackage[T1]{fontenc}
\usepackage{textcomp}
\usepackage{newtxtext}
\usepackage{newtxmath}
\usepackage{fix-cm} 
\usepackage{booktabs}
\usepackage{multirow}
\usepackage{float}
\usepackage{subcaption}
\usepackage{adjustbox}
\usepackage{rotating}
\usepackage{lscape}
\usepackage{colortbl}

\usepackage[utf8]{inputenc}
\usepackage{enumerate}
\usepackage{listings}
\usepackage{verbatim}
\usepackage{comment}
\usepackage{accents}
\usepackage{bigints}
\usepackage{arydshln}

\usepackage{setspace}
\usepackage{natbib}
\setcitestyle{authoryear}

\usepackage{xcolor}
\usepackage{hyperref}
\hypersetup{
  colorlinks=true,
  linkcolor=blue,
  citecolor=blue,
  urlcolor=blue
}

\usepackage[title]{appendix}

\newtheorem{theorem}{Theorem}[section]
\newtheorem{lemma}[theorem]{Lemma}

\newtheorem{proposition}{Proposition}[section]

\newtheorem{assumption}{Assumption}[section]
\numberwithin{equation}{section}
\newtheorem{condition}{Condition}[section]

\newlength{\dhatheight}

\newcommand{\anon}{1} 

\begin{document}

\def\spacingset#1{\renewcommand{\baselinestretch}{#1}\small\normalsize}
\spacingset{1}

\pagenumbering{arabic}\setcounter{page}{1}

\if1\anon
{
  \bigskip\bigskip\bigskip
  \begin{center}
    {\LARGE\bf Sparse VARs Do Not Imply Sparse Local Projections: Robust Inference for High-Dimensional Granger Causality}
  \end{center}
  \medskip

  \begin{center}
    Eug\`ene Dettaa\footnotemark[1]
    \hspace{1cm}
    Endong Wang\footnotemark[2]
  \end{center}

  \footnotetext[1]{D\'epartement de sciences \'economiques, Universit\'e de Montr\'eal, 3150 rue Jean-Brillant,
  Montr\'eal, H3T 1N8, Canada. Email: \href{mailto:eugene.delacroix.dettaa.mboudjiho@umontreal.ca}{eugene.delacroix.dettaa.mboudjiho@umontreal.ca}.
  Web: \url{https://eugenedettaa.github.io/}.}

  \footnotetext[2]{Department of Economics, University of Mannheim, L7, 3--5, Room 124, 68161 Mannheim, Germany.
  Email: \href{mailto:endong.wang@uni-mannheim.de}{endong.wang@uni-mannheim.de}.
  Web: \url{http://www.endongwang.com}.}

}
\fi

\if0\anon
{
  \bigskip\bigskip\bigskip
  \begin{center}
    {\LARGE\bf Sparse VARs Do Not Imply Sparse Local Projections: Robust Inference for High-Dimensional Granger Causality}
  \end{center}
  \medskip
}
\fi

\bigskip
\begin{abstract}
This paper studies multi-horizon Granger causality using high-dimensional local
projections in sparse Vector Autoregressive (VAR) systems. Since local
projection coefficients are nonlinear transformations of the underlying VAR
parameters, existing approaches, such as de-biased least absolute shrinkage and selection operator (LASSO) and
post-double-selection methods applied directly to local projections, lack a
general justification, as sparsity of the VAR does not always propagate to higher horizons.
We propose a two-step framework that avoids imposing sparsity at each horizon and delivers valid inference without
relying on heteroskedasticity- and autocorrelation-consistent (HAC) corrections.
We establish large sample theory for the proposed estimators and develop feasible Wald tests.
Monte Carlo experiments demonstrate improved size control across horizons relative to existing methods.
An application to large financial systems illustrates horizon-specific connectedness.
\end{abstract}



\vfill
\newpage
\spacingset{1.8} 

\section{Introduction}

This paper studies high-dimensional local projections (LPs) and develops Gaussian and robust inference for horizon-specific Granger-causal parameters in large dynamic systems. Local projections provide a transparent representation of multi-step dynamics and underpin a broad class of multi-horizon objects, including Granger-causality measures and network connectedness statistics\footnote{See, e.g., \cite{lutkepohl1993testing,dufour1998short,jorda2005estimation,dufour2006short,dufour2010short,dufour2012measuring,zhang2016exchange,diebold2014network,salamaliki2019transmission,diebold2023past}.}. Multi-horizon Granger causality is particularly relevant for connectedness analysis because spillovers may arise with delay and persist across horizons; focusing solely on one-step-ahead causality can therefore understate systemic interdependence, even though formal Granger-causality-based connectedness measures are largely confined to a single horizon \citep{billio2012econometric,brunetti2019interconnectedness,franch2024temporal}.

We depart from the existing literature along two dimensions. First, we impose sparsity only on the underlying data-generating VAR rather than separately at each projection horizon. Although LP estimands are numerically equivalent to those implied by the recursive VAR representation \citep{dufour1998short,wolf2021same}, the mapping from VAR slope parameters to horizon-specific LP coefficients is highly nonlinear and involves iterated dynamic propagation; consequently, even sparse VARs may induce dense LP coefficients at longer horizons, rendering horizon-by-horizon sparsity assumptions internally inconsistent. Second, we dispense with HAC-type corrections for serial correlation in projection errors. While standard LP inference relies on such corrections, HAC variance estimators are fragile and particularly unreliable at long horizons \citep{montiel2021local,xu2023local}. We instead develop inference procedures that remain valid without correcting for serial correlation.

Motivated by these considerations, we propose a two-step estimation framework for high-dimensional LPs. In the first step, we obtain regularized estimates of the underlying VAR that accommodate approximate sparsity. In the second step, we exploit closed-form mappings from VAR parameters to horizon-specific targets and apply a de-biasing correction to remove regularization-induced bias. This construction enables valid inference on multi-horizon parameters without imposing sparsity restrictions on the LP equations themselves. We illustrate the empirical relevance of the proposed methodology through multi-horizon Granger-causality tests in large financial systems, where horizon-specific restrictions yield a natural measure of system-wide connectedness.

This paper contributes to the literature on high-dimensional time series and multi-horizon Granger causality. To the best of our knowledge, we provide the first unified econometric framework for estimation and inference on horizon-specific Granger-causal effects in sparse high-dimensional VAR models.

First, building on the two-stage method of \citet{wang2024} developed for fixed-dimensional systems, we extend this framework to sparse high-dimensional VARs. Our approach relies on sparsity only at the level of the underlying VAR and does not impose sparsity on LP equations at each horizon. While existing results establish asymptotic properties in low-dimensional settings, their validity under high dimensionality has not been formally characterized. We close this gap by deriving large-sample theory for two-stage  estimators when the system dimension may exceed the sample size.

Second, we develop a robust inference theory for the resulting high-dimensional procedures. Under suitable regularity conditions, we establish asymptotic Gaussian inference and construct heteroskedasticity-consistent (HC) variance estimators that remain valid without HAC corrections. This approach avoids the finite-sample distortions associated with kernel and bandwidth selection in Newey--West-type inference and is computationally tractable in high-dimensional environments.

Third, we propose a de-biasing strategy for the two-stage estimators that builds on the de-sparsification framework of \citet{van2014asymptotically} and its extension to dynamic systems in \citet{adamek2023lasso}. The procedure combines regularized VAR estimation with Neyman-orthogonal score constructions \citep{chernozhukov2018double}, facilitating valid Gaussian inference for low-dimensional targets in high-dimensional settings.

We assess the finite-sample performance of Wald tests based on the proposed estimators through Monte Carlo experiments calibrated to empirically relevant designs. The results show that the two-stage procedure combined with HC variance estimation achieves accurate size control across all horizons, including long horizons where conventional HAC-based methods and post-double-selection approaches exhibit substantial distortions.

Finally, we illustrate the empirical relevance of our approach in an application to financial network connectedness. Focusing on major crisis episodes, the Lehman Brothers collapse, the Flash Crash, and the U.S. sovereign credit downgrade, we show that networks based on higher-horizon Granger-causality measures exhibit stronger and more persistent connectedness than networks constructed from one-step-ahead causality, highlighting the importance of longer-horizon transmission channels in assessing systemic risk.

\textit{Relevant literature.}
This paper contributes to the literature on regularized estimation in high-dimensional time series, including \cite{basu2015regularized,medeiros2016l1,wong2020lasso,masini2022regularized,adamek2023lasso}. We depart from this line of work by developing de-biased estimation and inference procedures for parameters in LP equations constructed from regularized VAR estimators. Our analysis builds on the desparsified estimation literature \cite{belloni2012sparse,van2014asymptotically,chernozhukov2018double,krampe2023structural} and complements recent work on single-horizon Granger causality in large systems \cite{hecq2023granger,babii2024high}. In contrast to \cite{adamek2024local}, who study de-biased inference for Sims impulse responses in high-dimensional LP models, we focus on multi-horizon Granger-causal parameters, following the conceptual distinction of \cite{dufour1998short}. We further extend the literature on robust inference in LP models \cite{montiel2021local,breitung2023projection,xu2024new,wang2024} to high-dimensional settings.

\textit{Outlines.}
Section~\ref{sec2frame} outlines the econometric framework. Section~\ref{sec5ts} presents the de-biased two-stage estimation method. Asymptotic Gaussian and robust inference results are derived in Section~\ref{sec6asym}. Monte Carlo evidence is reported in Section~\ref{sec7mc}, and an empirical application is provided in Section~\ref{sec8empirical}. Section~\ref{sec9con} concludes. Proofs are collected in the Online Appendix.

\textit{Notation.}
Throughout, $C$ denotes a generic positive constant whose value may change
from line to line, and $v_j$ denotes a conformable selection vector with one in the $j$-th position and zeros elsewhere. For a vector $x$, $\|x\|_1$ and $\|x\|_2$ denote the
$\ell_1$ and $\ell_2$ norms, respectively. For a matrix $B$, $\|B\|_1$,
$\|B\|_2$, and $\|B\|_\infty$ denote the induced matrix norms, and
$\|B\|_{\max}$ the elementwise maximum norm. For a symmetric matrix,
$\lambda_{\min}(\cdot)$ and $\lambda_{\max}(\cdot)$ denote its smallest and
largest eigenvalues.

\section{Econometric and financial motivations}
\label{sec2frame}

\subsection{Local projections for multi-horizon Granger causality}

Granger causality assesses whether the past of one variable improves the linear
prediction of another beyond its own history
\cite{granger1969investigating,geweke1984inference}. Extending the concept to
multi-step forecasts yields multi-horizon Granger causality
\cite{lutkepohl1993testing,dufour1998short}. Let
$w_t=(y_t,x_t,q_t')'$, where $x_t$ and $y_t$ are scalar variables of interest and
$q_t$ collects a possibly high-dimensional set of controls. Variable $x$ does not
Granger-cause $y$ at horizon $h$ if
\begin{align}
\label{gch}
\textup{P}_L\!\left( y_{t+h} \mid \mathcal{F}_{t} \right)
=
\textup{P}_L\!\left( y_{t+h} \mid \mathcal{F}_{-x,t} \right),
\end{align}
where $\textup{P}_L(\cdot\mid\cdot)$ denotes linear projection,
$\mathcal{F}_{t}=\sigma(w_t,w_{t-1},\ldots)$, and $\mathcal{F}_{-x,t}$ excludes the
history of $x$.

Suppose that $w_t$ follows a high-dimensional $\operatorname{VAR}(p)$,
\[
w_t = A_1 w_{t-1} + \cdots + A_p w_{t-p} + u_t,
\]
where $\{u_t\}$ is a martingale difference sequence with zero mean and nonsingular
covariance matrix $\Sigma_u$, satisfying $\lambda_{\min}(\Sigma_u)>0$. The lag
order $p$ is fixed, while the dimension $d$ may grow with the sample size $n$.

Under the VAR, the linear projection of $y_{t+h}$ on $\mathcal F_t$ admits the
local projection representation\footnote{See \cite{dufour1998short,wolf2021same}.}
$\textup{P}L(y_{t+h}\mid\mathcal F_t)=\beta_h' W_t$, where
$W_t=(w_t',w_{t-1}',\ldots,w_{t-p+1}')'$.
The coefficient vector $\beta_h$ corresponds to the row of the $h$-th power of
the VAR companion matrix $\mathbf A$ associated with $y$. Here,
$\mathbf A$ denotes the $dp\times dp$ companion matrix constructed from the VAR
slope matrices ${A_1,\ldots,A_p}$, so multi-horizon predictability is governed
by dynamic propagation through $\mathbf A^h$ (see, e.g.,
\cite[Section 2.2.6]{kilian2017structural}).

Partition $W_t=(W_{1,t}',W_{2,t}')'$, where
$W_{1,t}=(x_t,x_{t-1},\ldots,x_{t-p+1})'$ and $W_{2,t}$ contains all remaining
predictors. Then
\begin{align}
\label{lp_rew}
y_{t+h}
=
\beta_{1,h}' W_{1,t}
+
\beta_{2,h}' W_{2,t}
+
e_{t,h},
\end{align}
where $e_{t,h}$ is the projection error, which may be serially correlated. The
null hypothesis of Granger non-causality at horizon $h$ is
$\mathcal H_0:\beta_{1,h}=0$. Valid inference therefore requires estimating
$\beta_{1,h}$ in the presence of the high-dimensional nuisance parameter
$\beta_{2,h}$, even though sparsity of the VAR need not imply sparsity of
\eqref{lp_rew}.

\subsection{Sparsity does not propagate to local projections}

At horizon $h=1$, post--double-selection methods can be used to construct
de-biased estimators under the sparsity of the VAR and auxiliary regressions; see,
for example, \cite{hecq2023granger}. At longer horizons, however, feasibility
typically requires approximate sparsity of the local projection coefficients
themselves \cite{babii2024high}, an assumption that is generally untenable in
dynamic systems.

The key difficulty is that $\beta_h$ is determined by powers of the VAR companion
matrix $\mathbf A$, which are nonlinear functions of the primitive VAR
coefficients. As a result, sparse short-run VAR dynamics can generate dense
horizon-specific local projection coefficients as indirect transmission paths
accumulate over time. To illustrate this mechanism, consider a $d$-dimensional
VAR(1) with coefficient matrix $A$ satisfying
\begin{align}
\label{eq:bandA}
A_{ij}=\alpha\,\mathbbm{1}_{\{0\le j-i\le q\}},
\qquad |\alpha|<1,\quad q\ \text{fixed}.
\end{align}
Then $A$ is a sparse band matrix and is stable since $\rho(A)=|\alpha|<1$. However,
the horizon-$h$ coefficients correspond to $A^h$, whose support expands
monotonically with $h$ as longer transmission paths become active. In particular,
\begin{align}
\label{eq:support_growth}
\|v_1' A^h\|_0=\min\{d,qh+1\}.
\end{align}
where $\|\cdot\|_0$ denotes the cardinality of the support of a vector.
Thus, even when the VAR is sparse, the implied local projection
coefficients become dense as $h$ grows. This example
demonstrates that sparsity of the VAR does not propagate to
local projections, rendering horizon-by-horizon sparsity
assumptions inappropriate.

\subsection{Connectedness measures}

Existing Granger-causality-based measures of network connectedness with formal
testing are largely confined to a single horizon
\cite{billio2012econometric,brunetti2019interconnectedness}. By contrast,
\cite{diebold2014network} study multi-horizon connectedness using forecast error
variance decompositions without explicit inference. A
multi-horizon Granger-causality framework with formal testing therefore bridges
these two strands of the literature, while allowing for economically meaningful
horizon choices that range from short-term risk monitoring to longer-run
spillovers.

Let $p_{ij}(h)$ denote the $p$-value for testing the null that variable $j$ does
not Granger-cause variable $i$ at horizon $h$. For a given confidence level
$\alpha\in(0,1)$, define
\begin{align}
c_{ij}(h;\alpha)
=
\mathbf 1\!\left\{ p_{ij}(h)\le 1-\alpha \right\},
\qquad c_{ii}(h;\alpha)=0,
\end{align}
and summarize system-wide connectedness by the Degree of Granger Causality,
\begin{align}
\label{dgch}
\mathrm{DGC}(h;\alpha)
=
\frac{1}{N(N-1)}
\sum_{i=1}^{N}\sum_{j\neq i} c_{ij}(h;\alpha),
\end{align}
which measures the density of statistically significant predictive links and
takes values in $[0,1]$. Aggregating pairwise tests in this way, DGC is not
interpreted as a familywise error--controlled count of causal links, but as a
relative measure of network connectedness across horizons and subsamples. The use
of high confidence thresholds mitigates concerns about false discovery
accumulation, and the qualitative patterns are stable across alternative
thresholds, indicating systematic volatility spillovers rather than
noise-driven link formation.

\section{Two-stage estimation}
\label{sec5ts}

This section proposes a two-stage procedure for estimating the horizon-$h$
Granger causal parameters.

\subsection{Two-stage identification}

Suppose the reduced-form innovations $u_t$ satisfy a weak exogeneity condition
and the innovation covariance matrix $\Sigma_u$ is nonsingular. Let
$U_t=(u_t',u_{t-1}',\ldots,u_{t-p+1}')'$ denote the stacked innovation vector.
Under these conditions, the orthogonality restriction
$\mathrm P_L(y_{t+h}-\beta_h' W_t \mid U_t)=0$ holds, implying that $U_t$ is a valid
instrument set for $W_t$. Identification therefore follows from the moment
condition
\begin{align}
\label{eq:beta-ident}
\beta_h
=
\mathbb E[U_t W_t']^{-1}
\mathbb E[U_t y_{t+h}].
\end{align}
where the population moments admit closed-form expressions,
$\mathbb E[U_t W_t']=(I_p\otimes\Sigma_u)\Psi(p)$ and
$\mathbb E[U_t y_{t+h}]
=\mathbb E[U_t w_{t+h}']v_1
=(I_p\otimes\Sigma_u)
(\Psi_h',\Psi_{h+1}',\ldots,\Psi_{h+p-1}')'v_1$.\footnote{Since $y_t$ is the first element of $w_t$, $y_t=v_1'w_t$.}
Here $\Psi(p)$ is a $p\times p$ block matrix whose $(i,j)$ block equals
$\Psi_{i-j}'$ for $i\ge j$ and zero otherwise, and
$\Psi_h=J\mathbf A^h J'$. Since $\Sigma_u$ is full rank,
$\mathbb E[U_t W_t']$ is nonsingular, ruling out under-identification, as in the
static IV case.

To isolate the coefficient of interest, partition the instruments conformably
with the regressors. Define $U_{1,t}=R_1 U_t$ and $U_{2,t}=R_2 U_t$, where the
selection matrices $R_1$ and $R_2$ coincide with those used to define
$W_{1,t}=R_1 W_t=(x_t,x_{t-1},\ldots,x_{t-p+1})'$, and $W_{2,t}=R_2 W_t$ includes the rest of variables. Consequently, $U_{1,t}$ collects the VAR
innovations corresponding to $W_{1,t}$, and analogously for $U_{2,t}$ and
$W_{2,t}$. The Frisch--Waugh--Lovell theorem yields
\begin{align}
\label{4.15}
\beta_{1,h}
=
\mathbb E[U_{1,t}^{\perp} W_{1,t}']^{-1}
\mathbb E[U_{1,t}^{\perp} y_{t+h}],
\end{align}
where the residualized instrument
$U_{1,t}^{\perp}=U_{1,t}-\Gamma U_{2,t}$, with
$\Gamma=\mathbb E[U_{1,t} W_{2,t}']
\mathbb E[U_{2,t} W_{2,t}']^{-1}$,
satisfies the orthogonality condition
$\mathrm P_L(U_{1,t}^{\perp}\mid W_{2,t})=0$.
By construction, $U_{1,t}^{\perp}$ remains relevant for $W_{1,t}$ while being
orthogonal to the control regressors $W_{2,t}$.

\subsection{De-biased two-stage estimator}
This section develops a de-biased two-stage estimator for the multi-horizon
Granger-causal coefficient $\beta_{1,h}$ in the local-projection framework.
The construction is grounded in the population representation in \eqref{4.15}.
Identification is achieved through the residualized instrument
$U_{1,t}^{\perp}$, which partials out high-dimensional controls from the raw
instrument stack and isolates variation orthogonal to the nuisance component.

Let $\Sigma_{UW}:=\mathbb E[U_t W_t']$ denote the population cross-moment matrix
between the instrument vector $U_t$ and the regressor vector $W_t$. This matrix
captures the linear dependence between innovations and lagged regressors and
enters the definition of the partialling-out coefficient $\Gamma$. Using the
population expression for $\Gamma$, the residualized instrument admits the
closed-form representation
\begin{align}
\label{4.16_bis}
U_{1,t}^{\perp}
=
\bigl(R_1\Sigma_{UW}^{-1}R_1'\bigr)^{-1}
R_1\Sigma_{UW}^{-1}U_t,
\end{align}
which shows that $U_{1,t}^{\perp}$ depends only on second moments of $(U_t,W_t)$.

To implement \eqref{4.16_bis}, we exploit that $\Sigma_{UW}$ admits a closed-form
representation in terms of the reduced-form impulse responses
$\{\Psi_h\}_{h\ge 0}$ and the innovation covariance matrix $\Sigma_u$. We estimate
the VAR using a regularized estimator $\widehat{\mathbf A}_{1:p}^{(\mathrm{re})}$,
obtained by applying thresholding to Lasso-type estimates of the VAR slope
coefficients,
$\widehat{\mathbf A}_i^{(\mathrm{re})}
=
T_{\tau_A}(\widehat{\mathbf A}_i^{(\mathrm{lasso})})$,
with $\tau_A\asymp(\log(dp)/T)^{1/2}$, which ensures boundedness of the
column-wise $\ell_1$ norm as required by
Assumption~\ref{Assump_deLS_de2S}(\textit{iii}). Based on these estimates, we form
the fitted innovations
$\hat u_t := w_t-\widehat{\mathbf A}_{1:p}^{(\mathrm{re})}W_{t-1}$ and the sample
covariance estimator
$\widehat{\Sigma}_u := n^{-1}\sum_{t=p+1}^{n}\hat u_t\hat u_t'$. To stabilize covariance
estimation in high dimensions under sparsity of $\Sigma_u$
(Assumption~\ref{Assump_deLS_de2S}), we apply entrywise thresholding and define
$\widehat{\Sigma}_u^{(\mathrm{re})}:=T_{\tau_u}(\widehat{\Sigma}_u)$, where $T_{\tau_u}(\cdot)$
sets entries with absolute value below $\tau_u$ to zero. In theory, $\tau_u$ is
chosen on the order of $\|\widehat{\Sigma}_u-\Sigma_u\|_{\max}$
(see Section~\ref{sec6asym}); in applications, we set
$\tau_u\asymp(\log d/T)^{1/2}$.

Impulse responses are estimated from the regularized VAR. Let
$\widehat{\mathbf A}^{(\mathrm{re})}$ denote the associated companion matrix and
define $\hat\Psi_h := J(\widehat{\mathbf A}^{(\mathrm{re})})^h J'$. Substituting
$\{\hat\Psi_h\}$ and $\widehat{\Sigma}_u^{(\mathrm{re})}$ into the population
representation yields the plug-in estimator
\begin{align}
\widehat{\Sigma}_{UW}^{(\mathrm{re})}
=
\bigl(I_p\otimes\widehat{\Sigma}_u^{(\mathrm{re})}\bigr)\,\hat\Psi(p).
\end{align}

Given $\widehat{\Sigma}_{UW}^{(\mathrm{re})}$, the two-stage estimator of $\beta_{1,h}$
is obtained by instrumenting $W_{1,t}$ with residualized instruments
$\hat U_{1,t}^\perp$,
\begin{align}
\label{2S}
\hat \beta_{1,h}^{(2S)}
=
\bigg(\sum_t \hat U_{1,t}^{\perp} W_{1,t}'\bigg)^{-1}
\bigg(\sum_t \hat U_{1,t}^{\perp} y_{t+h}\bigg),
\end{align}
where $\hat U_{1,t}^{\perp}
=
\bigl(R_1(\widehat{\Sigma}_{UW}^{(\mathrm{re})})^{-1}R_1'\bigr)^{-1}
R_1(\widehat{\Sigma}_{UW}^{(\mathrm{re})})^{-1}\hat U_t,
$ $
\hat U_t=(\hat u_t',\ldots,\hat u_{t-p+1}')'.$ 
In high-dimensional settings, $\hat \beta_{1,h}^{(2S)}$ is generally not
$\sqrt n$-consistent, since the moment conditions are not Neyman-orthogonal with
respect to the high-dimensional nuisance parameter $\beta_{2,h}$.

We therefore apply a de-biasing correction and define
\begin{align}
\label{de-2s}
\hat \beta_{1,h}^{(\mathrm{de\text{-}2S})}
=
\hat \beta_{1,h}^{(2S)}
-
\bigg(\sum_t \hat U_{1,t}^{\perp} W_{1,t}'\bigg)^{-1}
\bigg(\sum_t \hat U_{1,t}^{\perp} W_{2,t}'\hat \beta_{2,h}\bigg),
\end{align}
where $\hat \beta_{2,h}$ is a consistent estimator of $\beta_{2,h}$ with
controlled bias. A convenient choice is the plug-in estimator from the
regularized VAR, $\hat \beta_{2,h}=J(\widehat{\mathbf A}^{(\mathrm{re})})^h$. This
correction removes the leading bias induced by projection onto $W_{2,t}$ and
restores first-order orthogonality, enabling valid inference.

\noindent\textbf{Remarks.}
\begin{enumerate}[(i)]

\item
We obtain regularized VAR slope estimates equation by equation via LASSO.
Specifically, for each $j=1,\ldots,d$, we solve
\begin{align*}
    \hat{A}_{j\bullet,1:p}^{(\mathrm{lasso})}
    =
    \arg\min_{A_{j\bullet,1:p}}
    \frac{1}{n-p}
    \sum_{t=p+1}^n
    \Bigg(
    w_{j,t}-\sum_{i=1}^p A_{j\bullet,i} w_{t-i}
    \Bigg)^2
    +
    \lambda \sum_{i=1}^p \|A_{j\bullet,i}\Pi_i\|_1 ,
\end{align*}
where $A_{j\bullet,1:p}$ denotes the $j$th row of the stacked VAR coefficient
matrices $A_{1:p}=(A_1,\ldots,A_p)$, and
$\Pi_i=\mathrm{diag}(\pi_{ik})_{k=1}^d$ collects penalty loadings. When $\pi_{ik}=1$ for all $i,k$, it reduces to the standard LASSO. Following \citet{belloni2012sparse}, we allow for data-dependent penalty loadings
to accommodate heterogeneity and self-normalization in the first-order
conditions.

\item
The covariance matrix $\Sigma_{UW}$ is generally not well approximated by the
naive sample covariance of $(\hat U_t,W_t)$ in high-dimensional settings, where $\hat U_t$ denotes stacked
least-squares VAR residuals. It is because such estimators may
be ill-conditioned or singular and fail to exploit the structural restrictions
implied by the VAR. In particular, $\Sigma_{UW}$ admits the structured
representation in \eqref{eq:beta-ident}, being lower triangular and block
Toeplitz, which the naive estimator does not impose. In finite samples, if the thresholded estimator $\widehat{\Sigma}_u^{(\mathrm{re})}$
is not full rank due to high dimensionality, we enforce nonsingularity via a
projection onto the space of positive definite matrices. This adjustment is
asymptotically innocuous under the maintained assumption that $\Sigma_u$ is
nonsingular, since with probability approaching one
$\widehat{\Sigma}_u^{(\mathrm{re})}$ lies in a neighborhood of $\Sigma_u$
contained in the interior of the positive-definite cone. Consequently, the
projection is locally smooth and the induced perturbation is
$o_p(n^{-1/2})$, leaving the first-order asymptotic expansion unchanged.

\item The de-biasing step is central to the construction of the proposed two-stage
estimator. As in high-dimensional least squares problems, regularization
introduces a bias that is non-negligible at the $\sqrt n$ scale and must be
accounted for to conduct valid inference. To see this, note that the two-stage
estimator admits the decomposition
\begin{align*}
\begin{split}
    \hat \beta_{1,h}^{(2S)}
    &=
    \bigg(\sum_t \hat U_{1,t}^{\perp} W_{1,t}' \bigg)^{-1}
    \bigg(\sum_t \hat U_{1,t}^{\perp} y_{t+h} \bigg) \\
    &=
    \beta_{1,h}
    +
    \bigg(\sum_t \hat U_{1,t}^{\perp} W_{1,t}' \bigg)^{-1}
    \bigg(\sum_t \hat U_{1,t}^{\perp} e_{t,h} \bigg)
    +
    \bigg(\sum_t \hat U_{1,t}^{\perp} W_{1,t}' \bigg)^{-1}
    \bigg(\sum_t \hat U_{1,t}^{\perp} W_{2,t}' \beta_{2,h} \bigg),
\end{split}
\end{align*}
where the final term represents the bias induced by estimation error in the
high-dimensional nuisance parameter $\beta_{2,h}$. This term generally does not
vanish at the $\sqrt n$ rate and motivates the subsequent de-biasing correction.

\item When $\hat U_{1,t}^{\perp}$ is constructed using the sample covariance estimator
of $\Sigma_{UW}$, the sample orthogonality condition
$\sum_t \hat U_{1,t}^{\perp} W_{2,t}'=0$ holds by construction. In high-dimensional
settings, however, $\hat U_{1,t}^{\perp}$ is computed using the population
representation of $\Sigma_{UW}$, so exact sample orthogonality is no longer
imposed and $n^{-1/2}\sum_t \hat U_{1,t}^{\perp} W_{2,t}'$ need not vanish. This
deviation is accommodated by the de-biasing correction and controlled via
Neyman orthogonality, leaving the $\sqrt n$ asymptotics unaffected.

\item The de-biased two-stage estimator
$\hat \beta_{1,h}^{(\mathrm{de\text{-}2S})}$ is defined as the solution to the
sample analogue of the moment condition, where the score function is
\[
\psi_t^{\mathrm{d2s}}\!\left(\beta_{1,h},\eta\right)
=
U_{1,t}^{\perp}
\bigl(
y_{t+h}-W_{1,t}'\beta_{1,h}-W_{2,t}'\beta_{2,h}
\bigr).
\]
The residualized instrument is defined as
$U_{1,t}^{\perp}:=U_{1,t}-\Gamma U_{2,t}$, with
$U_{1,t}=R_1 U_t$, $U_{2,t}=R_2 U_t$.
The nuisance vector
$\eta=\big(\beta_{2,h}',\mathrm{vec}(\Gamma)',\mathrm{vec}(\mathbf A)'\big)'$
collects the high-dimensional nuisance objects entering the score.
It is introduced as a convenient bundle of these objects, rather than as a
primitive parametrization of the VAR.
In particular, while $U_{1,t}^{\perp}$ is a random variable, its dependence on
the data is governed by nuisance functionals such as $\Gamma$, which are
themselves determined by the VAR primitives $(\mathbf A,\Sigma_u)$.
Accordingly, variations in the VAR primitives affect the score only through
these induced objects, and the score $\psi_t^{\mathrm{d2s}}$ satisfies Neyman
orthogonality with respect to $\eta$.

\item 
A key requirement for deriving the asymptotic distribution of the de-biased
two-stage estimator is that the bias arising from using estimated VAR residuals
$\hat u_t$ as instruments is asymptotically negligible. In the
low-dimensional setting, \citet{wang2024} establish that the estimation error
$(\hat u_t-u_t)$ does not contribute at the $\sqrt{n}$ order to the
two-stage estimator. In the subsequent section, we show that an analogous
property continues to hold in the high-dimensional framework considered here,
so that residual estimation error is asymptotically irrelevant for
first-order inference.

\end{enumerate}

\subsection{Newey--West--type HAC estimator}

We next derive feasible inference for the de-biased two-stage estimator by
characterizing its asymptotic distribution. The argument proceeds by isolating
the leading stochastic component in the $\sqrt{n}$-normalized estimation error
and verifying that all remaining terms are asymptotically negligible. We will show that the asymptotic behavior of the debiased two-stage
estimator is driven by a martingale-type sum involving the 
instrument $U_{t}$ and the projection error $e_{t,h}$. Importantly,
all effects stemming from first-stage regularization enter only through
higher-order remainder terms and therefore do not affect the limiting
distribution.

The asymptotic variance of
$\hat{\beta}_{1,h}^{(\mathrm{de\text{-}2S})}$ is given by
\begin{align}
\label{asymvar}
\operatorname{AVar}\!\left(
\sqrt{n}\hat{\beta}_{1,h}^{(\mathrm{de\text{-}2S})}
\right)
=
R_1 \Sigma_{UW}^{-1} 
\Omega_{U,h}
 \Sigma_{UW}^{\prime -1} R_1',
\end{align}
where
$\Omega_{U,h}
:=
\lim_{n\to\infty}
\operatorname{Var}\!\left(
n^{-1/2}\sum_t U_{t} e_{t,h}
\right)$
denotes the long-run variance of the score process.

For practical implementation, we construct a feasible estimator of the
asymptotic variance using a HAC procedure,
\begin{align}
\label{hac_de-2S}
\widehat{\operatorname{AVar}}^{(hac)}
\!\left(
\sqrt{n}\hat{\beta}_{1,h}^{(\mathrm{de\text{-}2S})}
\right)
=
R_1 \left(\widehat{\Sigma}_{UW}^{(\mathrm{re})}\right)^{-1} 
\hat\Omega_{U,h}^{(hac)}\left(\widehat{\Sigma}_{UW}^{(\mathrm{re})}\right)^{-1\prime} R_1',
\end{align}
where $\hat\Omega_{U,h}^{(hac)}$ is a consistent HAC estimator of
$\Omega_{U,h}$, such as the Newey--West estimator. In practice, the
latent instrument $U_{t}$ is replaced by its sample analogue
obtained from the sample analogue defined in \eqref{2S}. Moreover, since the
projection error $e_{t,h}$ is unobservable, we replace it by the residual
$\hat e_{t,h}=y_{t+h}-\hat\beta_h W_t$, where
$\hat\beta_h=v_1'(\widehat{\mathbf A}^{(\mathrm{re})})^h$. The resulting plug-in score
$\hat U_{t}\hat e_{t,h}$ is then used to construct the HAC covariance
estimator. Under standard regularity conditions, estimation error from the
first-stage regularization and the residualization step enters only through
higher-order terms and therefore does not affect first-order asymptotics. This
ensures the validity of the proposed HAC-based inference.

\subsection{Robust inference without serial correlation correction}
HAC-type covariance estimators are well known to perform poorly in finite samples,
particularly in long-horizon local projection settings, where they often induce
substantial size distortions. Their implementation further depends on
nontrivial kernel and bandwidth choices. These considerations motivate
alternative covariance estimators based solely on heteroskedasticity-robust
methods.

Replacing HAC estimation requires conditions under which serial dependence in
the score process can be eliminated. HAC estimators are primarily used to obtain
a positive semidefinite estimate of the long-run variance, a property that is
generally not preserved by naive aggregation of autocovariances. An alternative
approach is to construct a transformation of the score process that is
serially uncorrelated, so that the long-run variance coincides with the
contemporaneous covariance matrix. The long-run variance of the original score
process is the summation of all lead-lag autocovariances,
\begin{align}
\Omega_{U,h}
=
\sum_{k=-\infty}^{\infty}
\mathbb E\!\left[
U_t U_{t+k}' e_{t,h} e_{t+k,h}
\right].
\end{align}
The score process $U_t e_{t,h}$, where
$U_t=(u_t',u_{t-1}',\ldots,u_{t-p+1}')'$, is serially correlated, which motivates
the use of HAC corrections in standard inference. Following \cite{wang2024}, we
instead construct an alternative score sequence by re-indexing and stacking the
score components as
\begin{align}
\label{S_t}
s_t
:=
(e_{t,h},e_{t+1,h},\ldots,e_{t+p-1,h})'\otimes u_t .
\end{align}
This transformation aligns all score components with the contemporaneous
innovation $u_t$ and shifts serial dependence forward in time.

Under Assumption~\ref{cond_HC}, the transformed process $\{s_t\}$ is serially
uncorrelated. Consequently, the long-run variance of the original score process
coincides with the contemporaneous covariance matrix of $s_t$, namely
\begin{align}
    \sum_{j=-\infty}^{\infty}
\mathbb E\!\left[U_t e_{t,h}\, U_{t-j} e_{t-j,h}'\right]
=
\mathbb E[s_t s_t'] .
\end{align}
This identification ensures that the asymptotic variance appearing in the limit
distribution is consistently estimated by a heteroskedasticity-robust covariance
estimator, thereby obviating the need for HAC corrections.

\begin{assumption}
\label{cond_HC}
For all $t\ge1$:
\begin{enumerate}[(i)]
\item $\mathbb E[u_t\mid\{u_s\}_{s<t}]=0$ almost surely;
\item $\mathbb E[(u_tu_\tau')\otimes(u_{\tau+k}u_{\tau+k}')]=\mathbf 0$ for all
$\tau>t$ and $k>0$.
\end{enumerate}
\end{assumption}

Assumption~\ref{cond_HC}(i) imposes a martingale-difference structure on the
innovation process. Assumption~\ref{cond_HC}(ii) is a fourth-order orthogonality
condition: it requires that, for any $t<\tau$ and any forward offset $k>0$, the
cross-product $u_tu_\tau'$ is orthogonal (in the unconditional fourth moment) to
the future quadratic form $u_{\tau+k}u_{\tau+k}'$. This restriction is imposed
for variance estimation: it rules out precisely the type of intertemporal
fourth-moment dependence that would generate nonzero autocovariances in the
transformed score sequence $\{s_t\}$, so that a heteroskedasticity-consistent
variance estimator based only on contemporaneous score variation is valid.

Importantly, Assumption~\ref{cond_HC}(ii) does \emph{not} require serial
independence of $\{u_t\}$ and does \emph{not} rule out conditional
heteroskedasticity or volatility clustering. It allows the conditional covariance
matrix $\mathbb E[u_tu_t' \mid \{u_s\}_{s<t}]$ to vary over time, as long as this
variation does not induce the above cross--fourth-moment dependence across
nonoverlapping time blocks that would transmit serial dependence to $\{s_t\}$.
The condition is satisfied by a broad class of disturbance processes, including
i.i.d.\ shocks, mean-independent disturbances, conditionally homoskedastic
martingale differences, and conditionally Gaussian ARCH-type models (including
Gaussian GARCH specifications) in which the innovation is conditionally symmetric
and the relevant cross-products are orthogonal to future quadratic terms.

By contrast, Assumption~\ref{cond_HC}(ii) can be violated in environments where
second-moment dynamics feed back into future quadratic variation, e.g., through
leverage-type effects, which may induce serial dependence in the score sequence
and thus require HAC-type corrections. In the present framework,
Assumption~\ref{cond_HC} is imposed to rule out this source of score dependence.
In particular, it implies that the transformed score $s_t$ is serially
uncorrelated. The following proposition formalizes this implication.

\begin{proposition}
\label{proposition:uncorr}
Suppose Assumption~\ref{cond_HC} holds. Then
$\mathbb E[s_t s_\tau']=0$ for all $t\neq\tau$.
\end{proposition}

See Online Appendix for the proof. Accordingly, a
heteroskedasticity-robust estimator of the asymptotic covariance matrix for the
de-biased two-stage estimator is given by
\begin{align}
\label{hc_de-2S}
\widehat{\mathrm{AVar}}^{(hc)}
\!\left(\sqrt n\,\hat\beta_{1,h}^{(\mathrm{de\text{-}2S})}\right)
=
R_1\left(\widehat{\Sigma}_{UW}^{(\mathrm{re})}\right)^{-1}
\widehat{\mathrm{Var}}(\hat s_t)
\left(\widehat{\Sigma}_{UW}^{(\mathrm{re})}\right)^{'-1}
R_1',
\end{align}
where $\widehat{\mathrm{Var}}(\hat s_t)
=
\frac{1}{n-h}\sum_t \hat s_t \hat s_t'$, 
$\hat s_t
=
(\hat e_{t,h},\hat e_{t+1,h},\ldots,\hat e_{t+p-1,h})\otimes\hat u_t$,
$\hat u_t
=
w_t-\hat\Phi_{1:p}^{(re)}W_{t-1}$.
Here, $\hat e_{t,h}$ denotes the local-projection residual defined in
\eqref{2S}. Under standard regularity conditions, the estimator in
\eqref{hc_de-2S} is consistent for the asymptotic variance of
$\sqrt n\,\hat\beta_{1,h}^{(\mathrm{de\text{-}2S})}$.

\section{Asymptotic properties of two-stage estimators}
\label{sec6asym}

This section studies the large-sample behavior of the proposed de-biased two-stage
(de-2S) estimator. We proceed in two steps. First, we impose high-level assumptions for the regularized VAR estimator. Second, we establish asymptotic normality of the
de-2S estimator and consistency of feasible variance estimators. Both HAC and
heteroskedasticity-robust standard errors are covered.

\subsection{Assumptions}

We begin with conditions used in the preliminary consistency results. To justify
Lasso-type regularization for the VAR slope matrices
$\widehat{\mathbf A}^{(\mathrm{re})}_j$, $j=1,\ldots,p$, and their stacked form
$\widehat{\mathbf A}^{(\mathrm{re})}$, we impose approximate sparsity on the VAR
companion matrix. Following \cite{krampe2023structural} and
\cite{bickel2008covariance}, define the row-wise approximately sparse class
\begin{align}
\mathcal{U}(k,\mu)
=
\left\{
B=(b_{ij})\in\mathbb R^{r\times s}:
\max_{1\le i\le r}\sum_{j=1}^{s}|b_{ij}|^{\mu}\le k,
\ \|B\|_2\le C<\infty
\right\}.
\end{align}
This class includes exact sparsity as $\mu=0$ (interpreting
$\sum_{j=1}^s |b_{ij}|^\mu$ as the number of nonzero entries in row $i$) and
allows many small coefficients when $\mu\in(0,1)$, while controlling effective
row complexity through $k$.\footnote{In our application, $k$ may depend on
$(d,p)$ and is allowed to grow with $n$.}

Rates for $\ell_1$-regularized estimators in high-dimensional time series and
VARs are available under conditions of this type; see, e.g.,
\cite{basu2015regularized,adamek2023lasso}. We summarize the required inputs in a
high-level assumption.

\begin{assumption}\label{Assump_deLS_de2S}\leavevmode
\begin{enumerate}[(i)]
    \item \textup{(Row-wise and column-wise approximate sparsity)}
    $\mathbf{A} \in \mathcal{U}\left(k_A, \mu\right)$ and
    $\mathbf{A}^{\prime} \in \mathcal{U}\left(k_A, \mu\right)$ for some
    $\mu \in [0,1)$ and $k_A>0$.

    \item \textup{(Stability)} $\exists \,\varphi\in(0,1)$ such that $\forall \,m\in\mathbb N$,
    $\|\mathbf A^{m}\|_{2}\le C\varphi^{m}$ and
    $\|\mathbf A^{m}\|_{l}\le C k_A \varphi^{m}$ for $l\in\{1,\infty\}$.

    \item \textup{(Convergence rate of the Lasso-type regularized estimator)}
    $\widehat{\mathbf{A}}^{(\mathrm{r e})}$ satisfies
    \begin{align*}
    \left\|\widehat{\mathbf{A}}^{(\mathrm{r e})}-\mathbf{A}\right\|_l
    =
    O_p\left(
    k_A^{1.5}\left(\frac{\nu_n}{n}\right)^{(1-\mu) / 2}
    \right),
    \qquad l\in\{1,\infty\}.
    \end{align*}

    \item \textup{(Convergence rate of the sample covariance of innovations)}
    For all $U, V \in \mathbb{R}^{d \times d}$ with $\|U\|_2=1=\|V\|_2$,
    $$
    \bigg\|
    \frac{1}{n}\sum_{t=1}^n U\left(u_t u_t' -\Sigma_{u}\right)V
    \bigg\|_{\max }
    =
    O_p\left(\sqrt{\tilde{\nu}_n / n}\right).
    $$

    \item \textup{(Moment restrictions and weak dependence)}
    The VAR innovation process $\{u_t\}$ is $\alpha$-mixing with coefficients
    $\{\alpha(j)\}_{j\ge1}$ of size $r/(r-2)$ for some $r>2$, i.e.,
    $\sum_{j=1}^\infty j^{\frac{r}{r-2}-1}\alpha(j)<\infty,$
    and satisfies $\mathbb E|u_{i,t}|^{4r+\delta}\le q<\infty$ for all
    $i=1,\ldots,d$ and some $\delta>0$.

    \item \textup{(Nonsingularity and sparsity of the innovation covariance)}
    $\lambda_{\min}(\Sigma_u)\ge C>0$.
    Moreover, $\Sigma_u \in \mathcal U(k_U,\mu_u)$ for some $\mu_u\in[0,1)$ and
    $k_U>0$, and $\|\Sigma_u\|_\infty \le C k_U$.

    \item \textup{(Stability of the inverse covariance)}
    $\left\|\Sigma_{UW}^{-1}\right\|_{\infty} = O\!\left(k_{W}\right)$ for some
    $k_W>0$.

    \item \textup{(Fourth moments and long-run variance)}
    $\max_{1\le i\le d}\mathbb E|u_{i,t}|^{4}\le C,$
    $\max_{1\le i\le d}\mathbb E|w_{i,t}|^{4}\le C,$
    and the long-run variance matrix $\Omega_{U,h}$ satisfies
    $C^{-1}\le \lambda_{\min}(\Omega_{U,h})
    \le \lambda_{\max}(\Omega_{U,h}) \le C$.
\end{enumerate}
\end{assumption}

Assumption \ref{Assump_deLS_de2S}(i)--(ii) imposes approximate sparsity and
stability of the $dp\times dp$ companion matrix. The sparsity index $k_A$ may
grow with $(d,p)$, while stability yields geometric decay of $\mathbf A^m$ and
standard weak-dependence properties for the stacked state.

Assumption \ref{Assump_deLS_de2S}(iii) is a high-level rate for the regularized
estimator in $\|\cdot\|_1$ and $\|\cdot\|_\infty$, which are convenient for
bounding remainder terms. The factor $\nu_n$ collects dimension and tail effects
driving regularization error; for instance, one may take
$\nu_n=\log(dp)$ under sub-Gaussian tails, while heavier tails typically lead to
larger $\nu_n$.

Assumption~\ref{Assump_deLS_de2S}(iv) imposes a uniform concentration condition on
quadratic forms of the sample covariance estimator of $u_t$.
Equivalently, this requirement corresponds to an operator-norm concentration
bound for $\widehat{\Sigma}_u-\Sigma_u$. The condition does not impose sparsity on
$\Sigma_u$ itself and its effective dimensional complexity is summarized by the
factor $\tilde{\nu}_n$.  Assumption \ref{Assump_deLS_de2S}(v) provides a mixing and moment condition
sufficient for central limit and law-of-large-numbers arguments for the
quantities driving the de-biasing step. Assumption \ref{Assump_deLS_de2S}(vi)
ensures that $\Sigma_u$ is well conditioned and approximately sparse, which
supports thresholding-type estimation of $\Sigma_u$.

Assumption~\ref{Assump_deLS_de2S}(vii) controls the growth of
$|\Sigma_{UW}^{-1}|\infty$, which governs the sensitivity of the de-biasing
correction to estimation error in $\Sigma{UW}$.
This condition is imposed purely for inferential stability and is standard in
high-dimensional de-biasing arguments; see, for example,
Assumption~2(v) in \cite{krampe2023structural}. Finally,
Assumption \ref{Assump_deLS_de2S}(viii) ensures finite fourth moments and
uniform conditioning of the long-run variance $\Omega_{U,h}$; together with
(v), it underpins consistency of both HAC and heteroskedasticity-robust
variance estimators used for feasible inference.

\subsection{Theoretical results}

This subsection establishes asymptotic normality and feasible inference for the
de-biased two-stage estimator. The argument has two
ingredients. First, we derive high-dimensional consistency rates for the
covariance and cross-covariance estimators that enter the de-biasing correction.
Second, we impose growth conditions under which the resulting higher-order
remainder terms are asymptotically negligible, so that the studentized statistic
admits a standard Gaussian limit.

Under Assumption~\ref{Assump_deLS_de2S}, we obtain the following bounds for
$\widehat{\Sigma}_u$, $\widehat{\Sigma}_u^{(\mathrm{re})}$ and $\widehat{\Sigma}_{UW}^{(\mathrm{re})}$.

\begin{lemma}[Consistency results]\label{theo_consistency_deLS}
Under Assumption~\ref{Assump_deLS_de2S},
\begin{flalign*}
(i)\quad &
\|\widehat{\Sigma}_u-\Sigma_u\|_{\max}
=
O_p\Big(
\sqrt{\tilde{\nu}_n/n}
+
k_A^3(\nu_n/n)^{1-\mu}
+
k_A^{3/2}(\nu_n/n)^{(1-\mu)/2}\sqrt{\tilde{\nu}_n/n}
\Big)
=: \delta_n,
&\\
(ii)\quad&
\left\|\widehat{\Sigma}_u^{(\mathrm{re})}-\Sigma_u\right\|_{\infty}
=
O_p\big(k_U\delta_n^{1-\mu_u}\big);&\\
(iii)\quad&
\left\|\widehat{\Sigma}_{UW}^{(\mathrm{re})}-\Sigma_{UW}\right\|_{\infty}
=
O_p\Big(k_U
k_A^{3.5}\left(\nu_n/n\right)^{(1-\mu)/2}\delta_n^{1-\mu_u}+
k_Ak_U\delta_n^{1-\mu_u}
+
k_Uk_A^{3.5}\left(\nu_n/n\right)^{(1-\mu)/2} 
\Big).
\end{flalign*}
\end{lemma}
See Online Appendix for the proof. The rate $\delta_n$ collects two components. 
The first is the baseline high-dimensional sampling error
$\big\|n^{-1}\sum_{t=1}^n u_t u_t' - \Sigma_u\big\|_{\max}
= O_p(\sqrt{\tilde{\nu}_n/n})$ from Assumption~\ref{Assump_deLS_de2S}(\textit{iv}). 
The second is the plug-in error induced by estimating the VAR dynamics via
$\widehat{\mathbf A}^{(\mathrm{re})}$, which affects fitted innovations (and hence
$\widehat{\Sigma}_u^{(\mathrm{re})}$) and also propagates into the estimated impulse responses used to form
$\widehat{\Sigma}_{UW}^{(\mathrm{re})}$.
The second is sampling variability in covariance estimation, which would remain
even if the innovations $\{u_t\}$ were observed; this component is captured by
Assumption~\ref{Assump_deLS_de2S}(\textit{iv}). In particular, $\delta_n$ collects
both (a) the bias induced by using $\hat u_t$ in place of $u_t$ and (b) the
high-dimensional sampling error of the sample covariance of $\{u_t\}$. The bound
in part (iii) then combines the regularization error in
$\widehat{\Sigma}_u^{(\mathrm{re})}$ with the approximation error in the impulse
response coefficients $\hat{\Psi}_h$ constructed from powers of the regularized
transition matrix $(\widehat{\mathbf A}^{(\mathrm{re})})^h$.

As a benchmark, suppose $k_A$ and $k_U$ are fixed, the VAR slope matrices and
$\Sigma_u$ are exactly sparse (i.e., $\mu=\mu_u=0$), and the innovation process
has sub-Gaussian tails so that $\nu_n$ and $\tilde{\nu}_n$ are of order $\log(d)$.
Then both
$\|\widehat{\Sigma}_u^{(\mathrm{re})}-\Sigma_u\|_{\infty}$ and
$\|\widehat{\Sigma}_{UW}^{(\mathrm{re})}-\Sigma_{UW}\|_{\infty}$ achieve the
rate $O_p\big(\sqrt{\log(d)/n}\big)$.

\begin{condition}\label{cond1_de_2S}
\begin{flalign*}
(i)\quad &k_{W}\tilde \nu_n^{1/2}\Big(
k_U k_A^{3.5}\Big(\nu_n/n\Big)^{(1-\mu)/2}
+
k_A k_U\delta_n^{1-\mu_u}
\Big)=o(1);&\\
(ii)\quad&
k_{W}^2\Big(
k_U k_A^{3.5}\Big(\nu_n/n\Big)^{(1-\mu)/2}
+
k_A k_U\delta_n^{1-\mu_u}
\Big)=o(1);\\
(iii)\quad&
k_{W}^2\Big(
k_A^{4.5}\Big(\nu_n/n\Big)^{(1-\mu)/2}
\Big)=o(1).
\end{flalign*}
\end{condition}

Condition~\ref{cond1_de_2S} summarizes the growth restrictions needed for (a) the
$\sqrt{n}$-normalized de-biasing representation to be asymptotically linear and
(b) the plug-in standard error to be consistent. Part (i) enforces that the
dominant first-stage nuisance error entering the de-biasing correction vanishes
after scaling by $k_W$ and the stochastic fluctuation level $\tilde\nu_n^{1/2}$;
equivalently, it implies
$k_W \tilde\nu_n^{1/2}\|\widehat{\Sigma}_{UW}^{(\mathrm{re})}-\Sigma_{UW}\|_{\infty}
=o_p(1)$ up to the rate in Lemma~\ref{theo_consistency_deLS}(iii). Parts (ii) and
(iii) ensure that the estimation error in the asymptotic variance is negligible.
These restrictions control the effect of using
$\widehat{\Sigma}_{UW}^{(\mathrm{re})}$ and the feasible long-run variance
estimators $\widehat{\Omega}_{U,h}^{(hac)}$ or $\widehat{\Omega}_{U,h}^{(hc)}$ in
place of their population counterparts. In the benchmark case discussed above
with fixed $k_W$, the conditions reduce to standard requirements such as
$\log(d)/\sqrt{n}\to 0$ (and analogous rate restrictions implied by the remaining
terms).

We now state the main inference result.

\begin{theorem}[Inference for the de-2S estimator]
\label{theo_normality_de2S}
Under Assumptions \ref{Assump_deLS_de2S}, suppose Condition~\ref{cond1_de_2S}
holds. Then, for any $v\in\mathbb{R}^p$ with $\|v\|_1=1$,
\begin{equation}
\label{eq:de2S_studentized_hac}
\frac{\sqrt{n}\, v'(\hat \beta_{1,h}^{(\mathrm{de\text{-}2S})}-\beta_{1,h})}
{\widehat{s.e.}_{\hat\beta_{1,h}^{(\mathrm{de\text{-}2S})}}^{(\mathrm{hac})}(v)}
\xrightarrow{d} \mathcal{N}(0,1),
\end{equation}
where
$\widehat{s.e.}_{\hat\beta_{1,h}^{(\mathrm{de\text{-}2S})}}^{(\mathrm{hac})}(v)^2
:=
v'\widehat{\operatorname{AVar}}^{(\mathrm{hac})}\!\left(
\sqrt{n}\hat \beta_{1,h}^{(\mathrm{de\text{-}2S})}
\right)v$.
Moreover, if the VAR innovations $u_t$ satisfy Assumptions~\ref{cond_HC}, the same
limit result holds with
$\widehat{s.e.}^{(\mathrm{hac})}$ replaced by
$\widehat{s.e.}^{(\mathrm{hc})}$, where
$\widehat{s.e.}_{\hat\beta_{1,h}^{(\mathrm{de\text{-}2S})}}^{(\mathrm{hc})}(v)^2
:=
v'\widehat{\operatorname{AVar}}^{(\mathrm{hc})}\!\left(
\sqrt{n}\hat \beta_{1,h}^{(\mathrm{de\text{-}2S})}
\right)v$.
\end{theorem}

See Online Appendix for the proof. The HAC standard error is robust to general serial dependence in the innovation
process under the mixing and moment conditions in Assumption
\ref{Assump_deLS_de2S}. As an alternative, the heteroskedasticity-robust standard
error in \eqref{hc_de-2S} avoids long-run variance estimation but requires the
stronger conditions on $\{u_t\}$ summarized in Assumptions~\ref{cond_HC}.

\section{Monte Carlo simulations}
\label{sec7mc}

This section reports a Monte Carlo study assessing the finite-sample performance
of the proposed de-biased two-stage estimation and inference procedures for
high-dimensional local projections. The experiments focus on the size properties
of Wald tests for multi-horizon Granger-causal parameters under empirically
relevant departures from sparsity.

We consider vector autoregressive models of order $p=2$. Stationarity is imposed
via a factorization of the characteristic polynomial: two root matrices
$\{\Lambda_k\}_{k=1}^2$ are constructed and the VAR slope matrices are recovered as
$A_1=\Lambda_1+\Lambda_2$ and $A_2=-\Lambda_1\Lambda_2$, ensuring that all
eigenvalues of the companion matrix lie strictly inside the unit circle. The
innovations $\{u_t\}$ are i.i.d.\ Gaussian,
$u_t\sim\mathcal N(0,\Sigma_u)$, with $\Sigma_{u,ij}=0.5^{|i-j|}$. We consider
$(d,T)\in\{(60,120),(60,240)\}$. Under the VAR$(2)$ specification, the number of
regressors in the associated local projection equations is of the same order as
the sample size, rendering conventional OLS-based inference invalid.

Two designs for the root matrices are considered. In the first design, the roots
are tridiagonal, with $\Lambda_{ij,k}=\rho^{|i-j|+1}$ for $|i-j|\le q$, where
$q=3$ and $\rho=0.3$, and zero otherwise; the largest eigenvalue of the VAR companion
matrix equals $0.549$. In the second design, the roots are upper triangular, with
$\Lambda_{ij,k}=0.1$ if $0\le j-i\le q$ for $q=3$, and zero otherwise; the
corresponding largest eigenvalue of the VAR companion
matrix equals $0.147$.

Both designs are sparse at horizon one, but sparsity deteriorates rapidly as the
projection horizon increases due to dynamic propagation. In the tridiagonal
design, repeated matrix multiplication activates an expanding set of small but
non-negligible coefficients, rendering local projection equations increasingly
dense. In the upper-triangular design, the one-sided propagation structure induces
even faster support expansion, with the number of nonzero coefficients in the
$h$-step companion matrix growing at rate proportional to $qh$. As a result,
approximate sparsity breaks down at moderate horizons in both designs.

We conduct $1{,}000$ Monte Carlo replications. In each replication and for each
horizon $h=1,\ldots,24$, we compare four estimators: post--double-selection LASSO
with HAC inference; de-biased LASSO with HAC inference; the proposed de-biased
two-stage estimator with HAC inference; and the proposed de-biased two-stage
estimator with heteroskedasticity-robust inference. For all procedures relying on
HAC inference, long-run variance matrices are estimated using a Bartlett kernel
with bandwidth equal to the projection horizon $h$. Regularization parameters are
selected using BIC-type criteria.

Figures~\ref{figure_Tridiag}--\ref{figure_UpTri} report empirical rejection
frequencies of 5\% Wald tests for the two designs, computed as one minus the
empirical coverage probabilities of the corresponding 95\% confidence sets. Left
(right) panels correspond to $(d,T)=(60,120)$ ($(60,240)$). Ejection frequencies are based on the Wald statistic
\begin{align}
W_h
:=
(T-h)(\hat{\boldsymbol\beta}_{1,h}-\boldsymbol\beta_{1,h})'
\hat{\Omega}_h^{-1}
(\hat{\boldsymbol\beta}_{1,h}-\boldsymbol\beta_{1,h}),
\end{align}
where $\hat{\boldsymbol\beta}_{1,h}$ denotes the estimated vector of
Granger-causal coefficients at horizon $h$. The covariance matrix
$\hat{\Omega}_h$ is the corresponding asymptotic variance estimates. Under the null hypothesis,
$W_h \xrightarrow{d} \chi^2_q$, where $q=2$ is the number of restrictions in our VAR(2) simulations.

\begin{figure}[!ht]
    \centering
    \includegraphics[width=1 \textwidth]{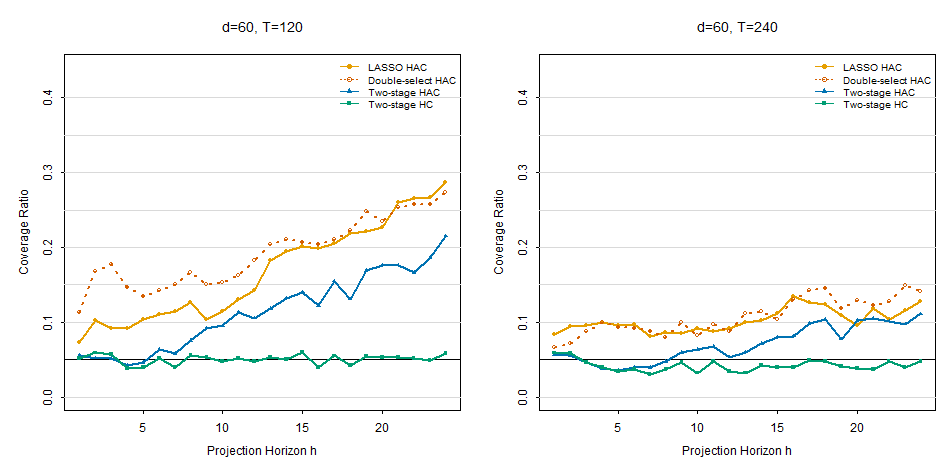}
    \caption{Empirical size of joint horizon-$h$ Granger-causality tests at the 5\% nominal level under the tridiagonal root design.}
    \label{figure_Tridiag}
\end{figure}

\begin{figure}[!ht]
    \centering
    \includegraphics[width=1 \textwidth]{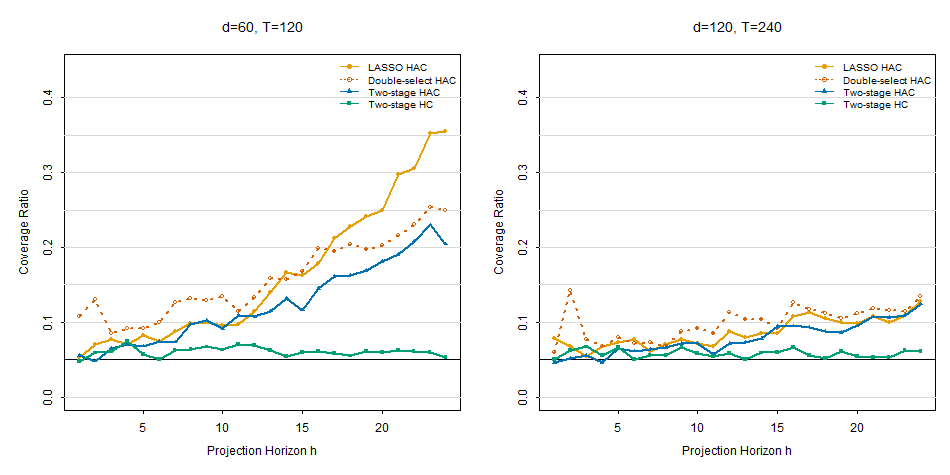}
    \caption{Empirical size of joint horizon-$h$ Granger-causality tests at the 5\% nominal level under the upper-triangular root design.}
    \label{figure_UpTri}
\end{figure}

Both de-biased LASSO and post--double-selection LASSO exhibit substantial size
distortions across horizons, reflecting violations of the sparsity conditions
required for valid de-biasing and variable selection. Although the VAR dynamics
are sparse at short horizons, local projection equations become progressively
dense as the horizon increases, placing these methods outside their intended
high-dimensional regimes.

Across both designs and sample sizes, the proposed de-biased two-stage estimator
combined with heteroskedasticity-robust inference delivers accurate and stable
size control, with rejection frequencies remaining close to the nominal level
uniformly across horizons. In contrast, procedures relying on HAC variance
estimation display increasing distortions at longer horizons, consistent with
documented finite-sample failures of HAC estimators in long-horizon local
projections \cite{montiel2021local,xu2023local,wang2024}. Increasing the sample
size mitigates distortions for all methods. Overall, the results indicate that
the proposed de-biased two-stage procedure provides reliable finite-sample
inference for multi-horizon Granger non-causality testing in high-dimensional
VARs.

\section{Empirical application}
\label{sec8empirical}

We apply the proposed methodology to study volatility transmission and network
connectedness in U.S.\ equity markets. A large literature documents that shocks
to firm-level volatility propagate across assets and sectors, generating
system-wide risk interdependence; see, among others,
\cite{diebold2014network,mcaleer2008realized,hecq2023granger,miao2023high}. Our
contribution is to characterize how such spillovers unfold across multiple
forecast horizons in a high-dimensional setting.

We use daily realized variance data for 30 large-cap U.S.\ equities,
corresponding to constituents of the Dow Jones Industrial Average and spanning
the main sectors of the U.S.\ economy.\if1\anon
{\footnote{We thank a colleague for providing the 10-minute realized variance data used in \cite{hecq2023granger}.}
}
\fi
\if0\anon
{\footnote{We thank a colleague for providing the data used in this study.}
}
\fi
Realized variances are constructed from intraday
10-minute returns and log-transformed to stabilize dispersion. The sample spans
March 2008 to February 2017, comprising 2{,}236 trading days.
\begin{figure}[ht]
    \centering
    \includegraphics[width=1 \textwidth]{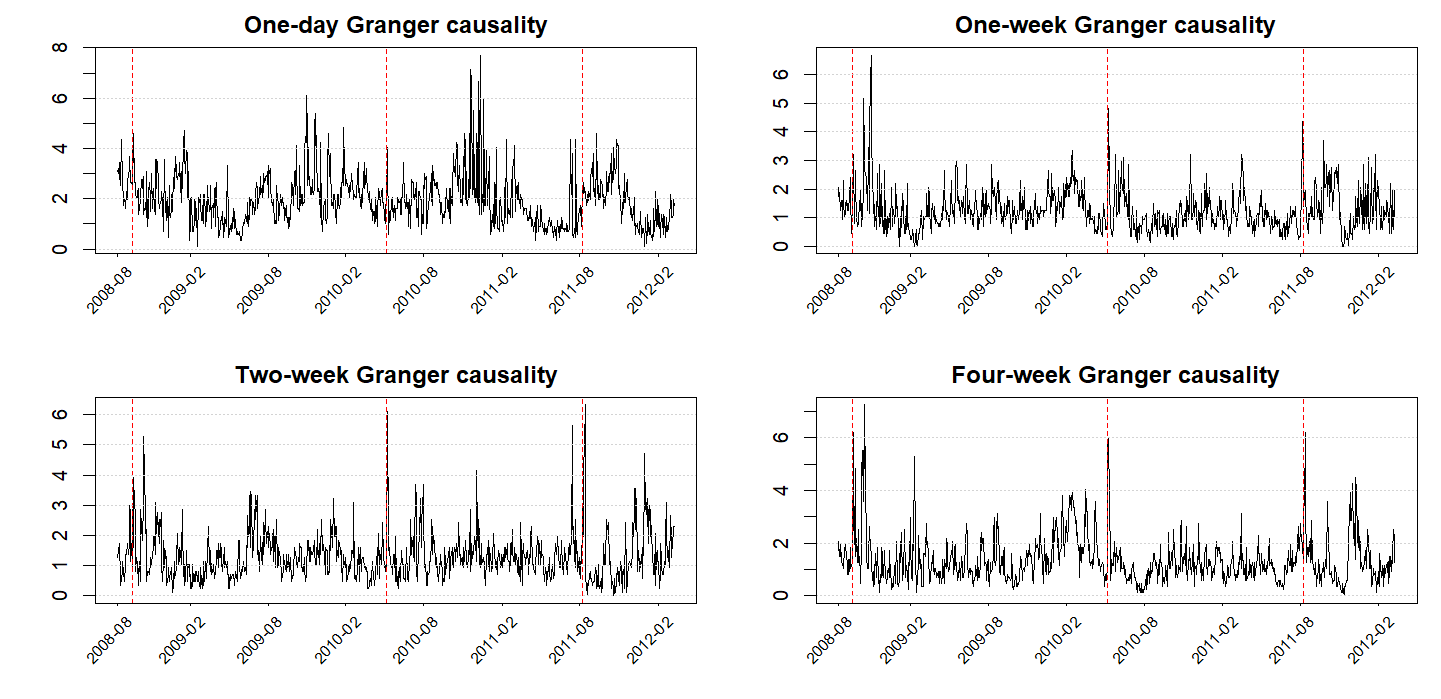}
    \caption{Rolling-window estimates (100 trading days) of total connectedness measured by DGC across horizons $h$ at $\alpha=0.99$.}
    \label{figure_TC99}
\end{figure}

\begin{figure}[ht!]
    \centering
    \includegraphics[width=1 \textwidth]{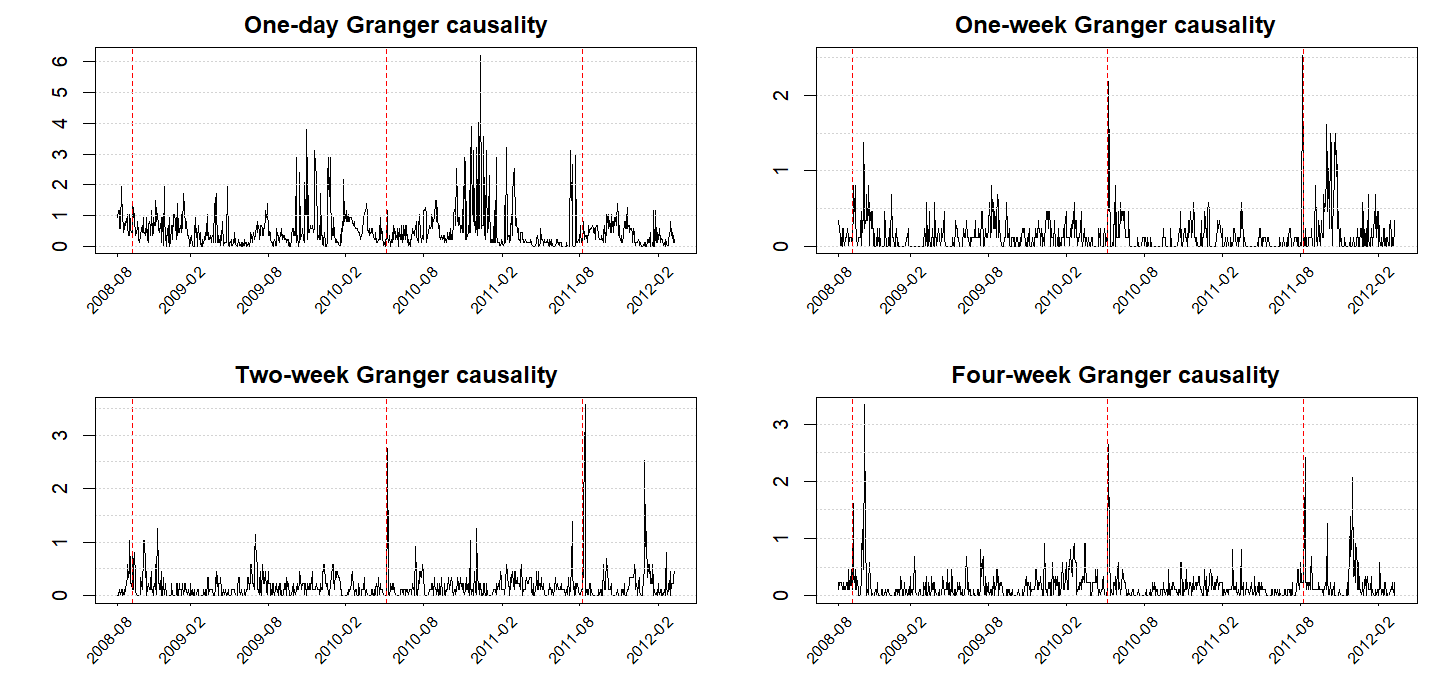}
    \caption{Rolling-window estimates (100 trading days) of total connectedness measured by DGC across horizons $h$ at $\alpha=0.999$.}
    \label{figure_TC999}
\end{figure}

We construct directed networks of predictive relations using multi-horizon
Granger-causality tests. For each horizon
$h\in\{1,5,10,20\}$ and each ordered pair $(i,j)$ with $i\neq j$, we test the null
that asset $j$ does not Granger-cause asset $i$ at horizon $h$. A directed link
is recorded whenever the null is rejected at a fixed significance level. The
resulting adjacency matrix captures the extensive margin of connectedness.
System-wide connectedness is summarized by the degree of Granger causality
$\mathrm{DGC}(h;\alpha)$ defined in \eqref{dgch}, evaluated at
$\alpha=0.99$ and $\alpha=0.999$.

\begin{figure}[ht]
    \centering
    \includegraphics[width=0.9 \textwidth]{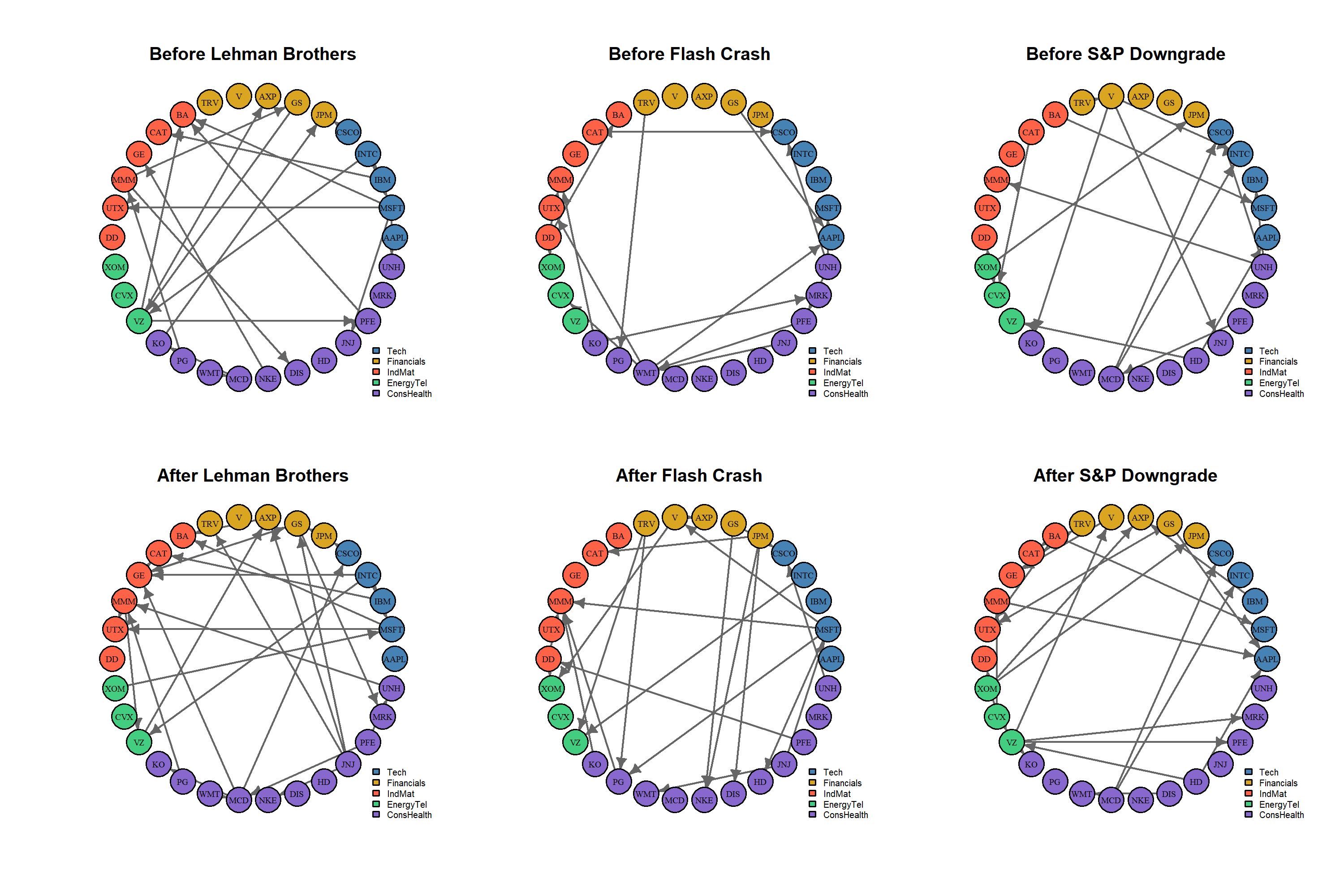}
    \caption{One-day Granger causality networks around major financial shocks. }
    \label{figure_totalbah1}
\end{figure}

\begin{figure}[ht!]
    \centering
    \includegraphics[width=0.9 \textwidth]{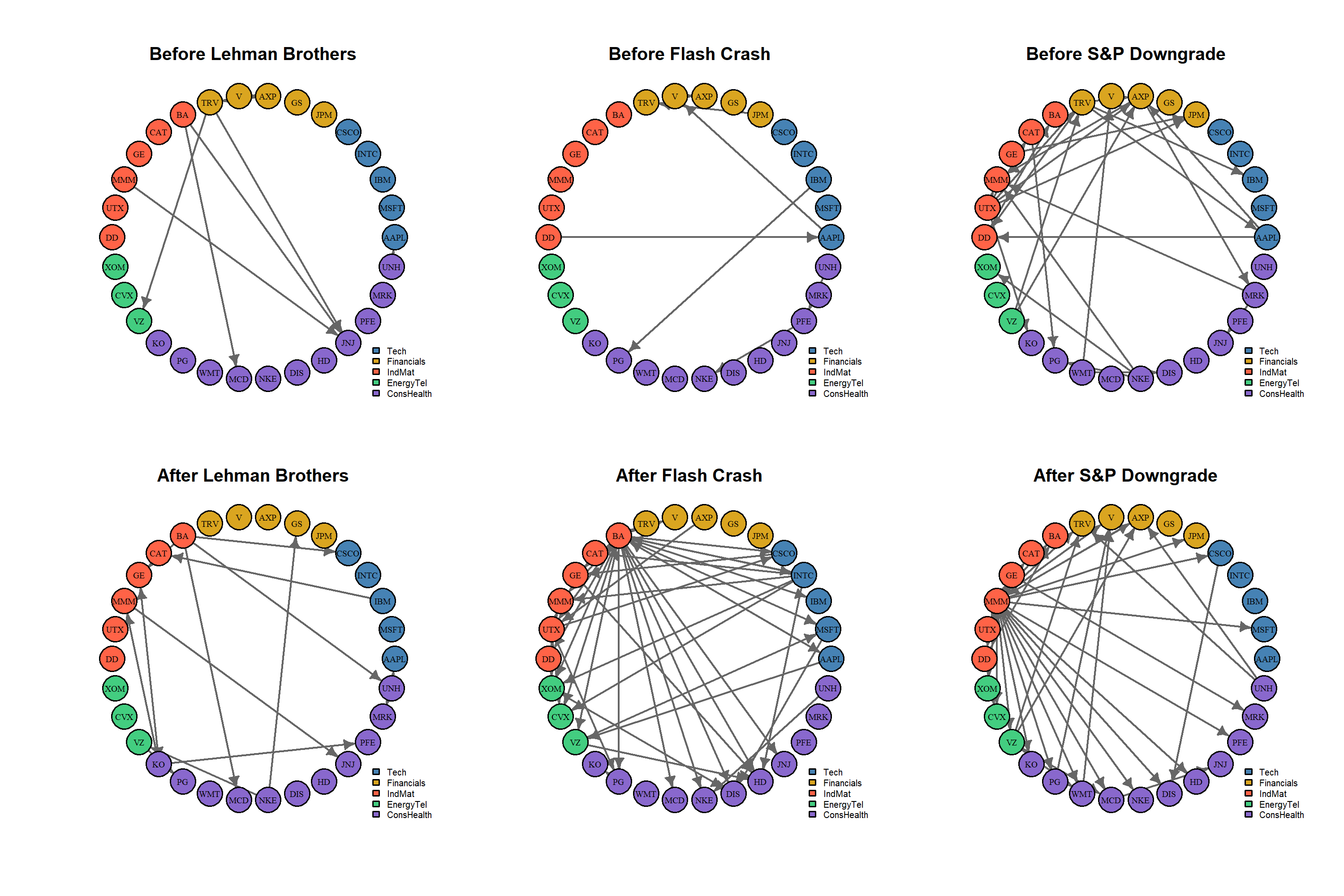}
    \caption{Two-week Granger causality networks around major financial shocks.}
    \label{figure_totalbah10}
\end{figure}

Figures~\ref{figure_TC99} and \ref{figure_TC999} report rolling-window estimates
of $\mathrm{DGC}(h;\alpha)$ based on 100-day windows. We focus on the period from
August 2008 to February 2012, covering three major market events: the Lehman
Brothers bankruptcy (2008--09--15), the Flash Crash (2010--05--06), and the S\&P
sovereign downgrade (2011--08--05), indicated by vertical lines. Within each
window, we estimate a sparse VAR(4) via adaptive LASSO and conduct
horizon-specific Granger-causality tests.

Two main patterns emerge. First, volatility connectedness is highly
state-dependent and rises sharply during periods of market stress. All three
events coincide with pronounced spikes in $\mathrm{DGC}(h;\alpha)$, indicating
a rapid densification of the volatility transmission network. Tightening the
significance level from $\alpha=0.99$ to $\alpha=0.999$ sharpens these spikes,
isolating episodes of acute systemic connectedness, while looser thresholds
capture more persistent but weaker dependence.

Second, the horizon dimension is central to identifying economically meaningful
spillovers. One-day-ahead connectedness responds only modestly around crisis
events, whereas multi-day horizons exhibit large
and persistent increases. This suggests that short-horizon measures understate
the importance of slower-moving propagation mechanisms, while longer horizons
capture cumulative volatility transmission. Notably, a pronounced one-day spike
in late 2010 is not associated with a major systemic event and disappears at
longer horizons, underscoring the value of multi-horizon analysis for filtering
spurious signals.

To visualize these dynamics, Figures~\ref{figure_totalbah1}--\ref{figure_totalbah10}
plot directed networks before and after each event at the one-day and two-week
horizons. Firms are grouped by sector and color-coded accordingly. While
one-day networks change little around events, two-week networks become markedly
denser following each shock, indicating broad-based amplification of volatility
spillovers across sectors. Our results highlight that equity-market volatility transmission
operates significantly through multi-day propagation channels. Crisis episodes are characterized by sharp increases in
predictive interdependence, with individual stocks becoming more informative
about future volatility dynamics of others. Monitoring horizon-specific
Granger-causal linkages therefore provides a financially meaningful and
complementary measure of market connectedness.

\section{Conclusion}
\label{sec9con}

We develop estimation and inference methods for multi-horizon Granger causality in high-dimensional dynamic systems. A key insight is that sparsity of the underlying VAR representation does not imply sparsity of local-projection coefficients at horizons $h>1$, because powers of the VAR transition matrix typically produce dense horizon-specific LP coefficients. This disconnect undermines direct application of high-dimensional regularization and post-selection methods to LP regressions.

We exploit sparsity at the VAR level and derive explicit analytic mappings to recover horizon-specific LP coefficients. Based on Neyman-orthogonal scores, we construct a de-biased two-stage estimator and establish asymptotic normality for low-dimensional multi-horizon Granger-causal parameters with growing system dimension. Our Wald tests are valid under approximate sparsity and weak dependence even as the number of variables increases with the sample size.

We also propose a heteroskedasticity-robust inference approach that obviates long-run variance estimation. By re-indexing score equations and imposing mild structural conditions on the innovation process, this route eliminates HAC corrections while preserving first-order validity. Monte Carlo evidence shows stable size control and improved finite-sample performance at long horizons relative to conventional HAC inference.

An empirical application to realized-volatility networks demonstrates that horizon-specific connectedness measures reveal propagation patterns that one-step analyses miss, particularly around market stress episodes. The findings underscore the value of multi-horizon causality measures for characterizing dynamic interdependence in large financial systems.

\if1\anon
\section*{Acknowledgments}
We are extremely grateful to Marine Carrasco, Benoit Perron, Mathieu Marcoux,
Jean-Marie Dufour, Carsten Trenkler, and Victoria Zinde-Walsh for helpful discussions and guidance.
We also thank Ren\'e Garcia, Prosper Dovonon, and participants in the
CIREQ Econometrics Conference in Honor of Eric Ghysels,
the 2024 NBER--NSF Time Series Conference,
the International Association for Applied Econometrics 2025 Conference,
and the Dagenais Econometrics Seminars
for constructive comments.
This work was supported by the Fonds de recherche du Qu\'ebec -- Soci\'et\'e et culture (FRQSC)
and the research funds provided by the University of Mannheim.
\fi
\bibliographystyle{agsm}
\bibliography{references}

\clearpage
\newpage
\pagenumbering{arabic}\setcounter{page}{1}
\appendix

\setcounter{section}{0}
\renewcommand{\thesection}{A.\arabic{section}}
\renewcommand{\thesubsection}{A.\arabic{section}.\arabic{subsection}}

\renewcommand{\thefigure}{A.\arabic{figure}}
\setcounter{figure}{0}
\renewcommand{\thetable}{A.\arabic{table}}
\setcounter{table}{0}

\setcounter{theorem}{0}
\renewcommand{\thetheorem}{A.\arabic{theorem}}

\setcounter{proposition}{0}
\renewcommand{\theproposition}{A.\arabic{proposition}}

\makeatletter
\renewenvironment{proof}[1][\proofname]{\par\pushQED{\qed}\normalfont
  \topsep6\p@\@plus6\p@\relax
  \trivlist\item[\hskip\labelsep\itshape #1.]\ignorespaces
}{\popQED\endtrivlist\@endpefalse}
\makeatother

\begin{center}
{\Large \bf Online Appendix}

\vspace{0.3cm}

{\large to ``Sparse VARs Do Not Imply Sparse Local Projections:
Robust Inference for High-Dimensional Granger Causality''}

\vspace{0.45cm}
\if1\anon
{
{\large Eug\`ene Dettaa \qquad Endong Wang}
}
\fi
\if0\anon
{
}
\fi

\end{center}

\section{Guide to the online appendix}\label{sec:OA_guide}

This online appendix is organized as follows.
Section~\ref{sec:OA_main} presents the proofs of the main results in the order they appear in the paper.
Section~\ref{sec:OA_lemmas} gathers auxiliary lemmas used repeatedly in the main proofs.
Section~\ref{sec:OA_cltvar} collects the technical proofs for lemmas. For notational convenience, we omit the superscript \textup{(re)} for regularized estimators throughout the appendix.
For example, $\widehat{\mathbf A}$ denotes $\widehat{\mathbf A}^{\textup{(re)}}$.
The symbol $C$ denotes a generic positive constant whose value may change from line to line.

We use the following abbreviations: T (triangle inequality), CS (Cauchy--Schwarz inequality),
LIE (law of iterated expectations), and m.d.s.\ (martingale difference sequence).

We repeatedly use the following elementary inequalities. For any compatible matrices $B_1,B_2,U,V$ and any vector $x$,
\begin{align*}
& \|B_1'\|_{\max}=\|B_1\|_{\max},
\qquad \|B_1'\|_{\infty}=\|B_1\|_{1}, \\
& \|B_1B_2\|_{\max}\le \|B_1\|_{\infty}\|B_2\|_{\max},
\qquad \|B_1x\|_{\ell}\le \|B_1\|_{\ell}\|x\|_{\ell}, \\
& \|B_1B_2\|_{\ell}\le \|B_1\|_{\ell}\|B_2\|_{\ell},
\qquad \ell\in\{1,2,\infty\},\\
& \|UB_1V\|_2=\|B_1\|_2 \quad \text{if } \|U\|_2=\|V\|_2=1 .
\end{align*}

\section{Proofs of main results}\label{sec:OA_main}

\begin{proof}[\textbf{Proof of Proposition} \ref{proposition:uncorr}]
\label{proof:proposition:uncorr}
We show that Assumption \ref{cond_HC} implies that $\{s_t\}$ is serially uncorrelated. Recall that
$
s_t=\big(e_{t,h},e_{t+1,h},\ldots,e_{t+p-1,h}\big)\otimes u_t,
$
so for $t\neq \tau$,
\begin{align}
\mathbb E[s_t s_\tau']
&=\mathbb E\!\left[\big(e_{t,h},\ldots,e_{t+p-1,h}\big)\big(e_{\tau,h},\ldots,e_{\tau+p-1,h}\big)'\otimes u_tu_\tau'\right].
\end{align}
Hence, $\mathbb E[s_t s_\tau']=\mathbf 0$ is equivalent to
\begin{align}
\label{eq:uncorr_reduce}
\mathbb E\!\left[u_tu_\tau' \otimes e_{t+i,h}e_{\tau+j,h}\right]=\mathbf 0,
\qquad i,j=0,1,\ldots,p-1.
\end{align}

Next, use the linear representation $e_{t,h}=\sum_{m=1}^h v_1'\Psi_{h-m}u_{t+m}$. For any $i,j\in\{0,\ldots,p-1\}$,
\begin{align}
\label{eq:expand_e}
\mathbb E\!\left[u_tu_\tau' \otimes e_{t+i,h}e_{\tau+j,h}\right]
&=\sum_{m=1}^h\sum_{n=1}^h 
\mathbb E\!\left[u_tu_\tau' \otimes \Big(v_1'\Psi_{h-m}u_{t+i+m}u_{\tau+j+n}'\Psi_{h-n}'v_1\Big)\right].
\end{align}

Without loss of generality, assume $t<\tau$ and fix $(m,n,i,j)$. Consider three cases. First, if $t+i+m=\tau+j+n$, then $t+i+m>\tau>t$, and Assumption \ref{cond_HC}(ii) yields
\begin{align}
\mathbb E\!\left[u_tu_\tau' \otimes \Big(v_1'\Psi_{h-m}u_{t+i+m}u_{\tau+j+n}'\Psi_{h-n}'v_1\Big)\right]=0.
\end{align}
Second, if $t+i+m<\tau+j+n$, then $\tau+j+n>\max\{t+i+m,\tau,t\}$, and Assumption \ref{cond_HC}(i) together with LIE implies the same term equals $0$. Third, if $t+i+m>\tau+j+n$, then $t+i+m>\tau+j+n>\tau>t$, and again Assumption \ref{cond_HC}(i) and LIE imply the term equals $0$.

In all cases, each summand in \eqref{eq:expand_e} is zero, so \eqref{eq:uncorr_reduce} holds, and therefore $\mathbb E[s_t s_\tau']=\mathbf 0$ for all $t\neq \tau$. Hence $\{s_t\}$ is serially uncorrelated. It follows that the long-run variance
$
\Omega_{U,h}:=\sum_{k=-\infty}^{\infty}\mathbb E[s_t s_{t+k}']
$
reduces to the variance,
$
\Omega_{U,h}=\mathbb E[s_t s_t']=\mathrm{Var}(s_t).
$
\end{proof}

\begin{proof}[\textbf{Proof of Lemma \ref{theo_consistency_deLS}}]
We first establish the rate in part (i). Write
\begin{align}\label{eq:Sigma_u_decomp}
\begin{split}
\big\|\widehat{\Sigma}_u-\Sigma_u\big\|_{\max}
&=
\left\|
\frac{1}{n-p}\sum_{t=p+1}^n \hat u_t \hat u_t' - \Sigma_u
\right\|_{\max} \\
&\le
\left\|
\frac{1}{n-p}\sum_{t=p+1}^n (\hat u_t-u_t)(\hat u_t-u_t)'
\right\|_{\max} \\
&\quad+2
\left\|
\frac{1}{n-p}\sum_{t=p+1}^n (\hat u_t-u_t)u_t'
\right\|_{\max} +
\left\|
\frac{1}{n-p}\sum_{t=p+1}^n u_tu_t'-\Sigma_u
\right\|_{\max}.
\end{split}
\end{align}
Hence,
\[
\|\widehat{\Sigma}_u-\Sigma_u\|_{\max}\le I_1+2I_2+I_3,
\]
where $I_1,I_2,I_3$ denote the four terms above in obvious order. By Assumption
\ref{Assump_deLS_de2S}(iv), $I_3=O_p\big(\sqrt{\tilde{\nu}_n/n}\big)$.

Using $u_t=w_t-J\mathbf A W_{t-1}$ and $\hat u_t=w_t-J\widehat{\mathbf A}W_{t-1}$, we have
$\hat u_t-u_t=J(\mathbf A-\widehat{\mathbf A})W_{t-1}$. Lemma \ref{lemma1} implies
\[
\Big\|
\frac{1}{n-p}\sum_{t=p+1}^n W_{t-1}W_{t-1}'
\Big\|_{\max}
=O_p(1),
\qquad
\Big\|
\frac{1}{n-p}\sum_{t=p+1}^n W_{t-1}u_t'
\Big\|_{\max}
=O_p\big(\sqrt{\tilde{\nu}_n/n}\big).
\]

For $I_1$, using submultiplicativity and $\|J\|_\infty=1$,
\[
I_1
\le
\|\widehat{\mathbf A}-\mathbf A\|_\infty
\|\widehat{\mathbf A}-\mathbf A\|_1
\Big\|
\frac{1}{n-p}\sum_{t=p+1}^n W_{t-1}W_{t-1}'
\Big\|_{\max}
=
O_p\big(\|\widehat{\mathbf A}-\mathbf A\|_\infty^2\big),
\]
where Assumption \ref{Assump_deLS_de2S}(iii) is used. Similarly,
\[
I_2
\le
\|\widehat{\mathbf A}-\mathbf A\|_\infty
\Big\|
\frac{1}{n-p}\sum_{t=p+1}^n W_{t-1}u_t'
\Big\|_{\max}
=
O_p\left(\|\widehat{\mathbf A}-\mathbf A\|_\infty\sqrt{\tilde{\nu}_n/n}\right).
\]

Combining the above bounds yields
\[
\|\widehat{\Sigma}_u-\Sigma_u\|_{\max}
=
O_p\left(
k_A^3(\nu_n/n)^{1-\mu}
+ k_A^{1.5}(\nu_n/n)^{(1-\mu)/2}\sqrt{\tilde{\nu}_n/n}+
\sqrt{\tilde{\nu}_n/n}
\right)
=:O(\delta_n).
\]

The thresholded estimator $\widehat{\Sigma}_u^{(\mathrm{re})}=T_{\tau_u}(\widehat{\Sigma}_u)$
is constructed with $\tau_n\asymp\delta_n$. Under Assumption
\ref{Assump_deLS_de2S}(vi), $\Sigma_u\in\mathcal U(k_U,\mu_u)$. For each row $i$,
\[
\sum_{j=1}^d|\widehat{\Sigma}^{(\mathrm{re})}_{u,ij}-\sigma_{u,ij}|
\le
\sum_{j:\,|\widehat{\Sigma}_{u,ij}|\ge\tau_n}|\widehat{\Sigma}_{u,ij}-\sigma_{u,ij}|
+
\sum_{j:\,|\widehat{\Sigma}_{u,ij}|<\tau_n}|\sigma_{u,ij}|
\lesssim
k_U\tau_n^{1-\mu_u}.
\]
Consequently,
\[
\|\widehat{\Sigma}_u^{(\mathrm{re})}-\Sigma_u\|_\infty
=
O_p\big(k_U\delta_n^{1-\mu_u}\big).
\]
If a positive-definite adjustment is applied to enforce nonsingularity,
this rate is preserved under $\lambda_{\min}(\Sigma_u)\ge C>0$, since the
adjustment only modifies $\widehat{\Sigma}_u^{(\mathrm{re})}$ on a set where
its smallest eigenvalue is below a fixed constant and does not amplify the
entrywise estimation error. In particular, all subsequent bounds involving
$(\widehat{\Sigma}_u^{(\mathrm{re})})^{-1}$ remain valid.

The proof of part (ii) relies on auxiliary bounds for $\hat\Psi_j-\Psi_j$ and
$\sum_{m\ge0}(\hat\Psi_m-\Psi_m)$. Under Assumption \ref{Assump_deLS_de2S}(ii),
$\sum_{m=0}^{\infty}\|\mathbf A^m\|_{\infty}=O(k_A)$. By Assumption
\ref{Assump_deLS_de2S}(iii),
$\|\widehat{\mathbf A}-\mathbf A\|_{\infty}
=O_p\big(k_A^{1.5}(\nu_n/n)^{(1-\mu)/2}\big)$, which implies
$k_A\|\widehat{\mathbf A}-\mathbf A\|_{\infty}=o_p(1)$. Together with stability of
$\mathbf A$, this ensures that, with probability approaching one, $\widehat{\mathbf A}$
is stable and satisfies $\|\widehat{\mathbf A}^m\|_{\infty}\le Ck_A\tilde\varphi^{\,m}$
for all $m\ge0$ and some $\tilde\varphi\in(\varphi,1)$. Consequently,
$\sum_{m=0}^{\infty}\|\widehat{\mathbf A}^m\|_{\infty}=O_p(k_A)$.

Let $j\in\mathbb N$ be fixed and sample-independent. Using the identity
$\widehat{\mathbf A}^j-\mathbf A^j
=\sum_{s=0}^{j-1}\widehat{\mathbf A}^{s}(\widehat{\mathbf A}-\mathbf A)\mathbf A^{j-1-s}$,
together with $\hat\Psi_j=J\widehat{\mathbf A}^jJ'$, $\Psi_j=J\mathbf A^jJ'$ and
$\|J\|_\infty=1$, we obtain
\begin{align}
\begin{split}
\|\hat\Psi_j-\Psi_j\|_\infty
&\le
\|\widehat{\mathbf A}-\mathbf A\|_\infty
\sum_{s=0}^{j-1}\|\widehat{\mathbf A}^{s}\|_\infty
\|\mathbf A^{j-1-s}\|_\infty  \\
&=
O_p\big(k_A^2\|\widehat{\mathbf A}-\mathbf A\|_\infty\big)
=
O_p\left(k_A^{3.5}\left(\nu_n/n\right)^{(1-\mu)/2}\right).
\end{split}
\end{align}

The same argument applied to $\mathbf A'$ yields
$\sum_{m=0}^{\infty}\|(\mathbf A')^m\|_{\infty}=O(k_A)$ and 
$\|\hat\Psi_j'-\Psi_j'\|_{\infty}
=O_p\big(k_A^{3.5}(\nu_n/n)^{(1-\mu)/2}\big)$.

Using the explicit representation
$\Sigma_{UW}$, we have
\[
\|\widehat{\Sigma}_{UW}^{(\mathrm{re})}-\Sigma_{UW}\|_{\infty}
\lesssim
\sum_{j=0}^{p-1}
\|\widehat{\Sigma}_u^{(\mathrm{re})}\hat\Psi_j'-\Sigma_u\Psi_j'\|_{\infty}.
\]
Assumption \ref{Assump_deLS_de2S}(v) implies $\|\Sigma_u\|_\infty=O(k_U)$, and part (i)
yields
$\|\widehat{\Sigma}_u^{(\mathrm{re})}-\Sigma_u\|_\infty
=O_p\big(k_U\delta_n^{1-\mu_u}\big)$. Since $p$ is fixed, it suffices to bound a
generic summand, for which
\begin{align}
\begin{split}
\big\|\widehat{\Sigma}_u^{(\mathrm{re})}\hat\Psi_j'
-
\Sigma_u\Psi_j'\big\|_{\infty}
\;&\le\;
\big\|\widehat{\Sigma}_u^{(\mathrm{re})}-\Sigma_u\big\|_{\infty}
\big\|\hat\Psi_j-\Psi_j\big\|_{1} \\
&\quad+
\big\|\widehat{\Sigma}_u^{(\mathrm{re})}-\Sigma_u\big\|_{\infty}
\big\|\Psi_j\big\|_{1}
+
\big\|\Sigma_u\big\|_{\infty}
\big\|\hat\Psi_j-\Psi_j\big\|_{1},
\end{split}
\end{align}

Substituting the above bounds gives
\begin{align}
\begin{split}
\|\widehat{\Sigma}_{UW}^{(\mathrm{re})}-\Sigma_{UW}\|_{\infty}
=
O_p\Big(
&k_A^{3.5}\left(\nu_n/n\right)^{(1-\mu)/2}k_U\delta_n^{1-\mu_u}+
k_Ak_U\delta_n^{1-\mu_u}
+
k_Uk_A^{3.5}\left(\nu_n/n\right)^{(1-\mu)/2} 
\Big).
\end{split}
\end{align}

\end{proof}
\begin{proof}[\textbf{Proof of Theorem \ref{theo_normality_de2S}}]

Let $G := R_1 \Sigma_{UW}^{-1}$. By Assumption \ref{Assump_deLS_de2S}(\textit{vii}),
$\Omega_{U,h}$ is positive definite with $\lambda_{\min}(\Omega_{U,h}) \ge C^{-1}$.
For any nonzero $v\in\mathbb R^p$,
\begin{align}\label{eq:lower2}
\begin{split}
v' G \Omega_{U,h} G' v
&\ge
\lambda_{\min}(\Omega_{U,h}) \, \|G'v\|_2^2  \\
&=
\lambda_{\min}(\Omega_{U,h}) \, \|\Sigma_{UW}^{\prime-1} R_1' v\|_2^2 \\
&\ge
\lambda_{\min}(\Omega_{U,h}) \, \sigma_{\min}^2(\Sigma_{UW}^{-1}) \, \|R_1'v\|_2^2 .
\end{split}
\end{align}
Since $R_1R_1'=I_p$, $\|R_1'v\|_2=\|v\|_2$. Moreover,
Assumption \ref{Assump_deLS_de2S}(\textit{vii}) implies
$\sigma_{\min}(\Sigma_{UW}^{-1})\ge C^{-1}$. Hence
$v' G \Omega_{U,h} G' v \ge C^{-1}\|v\|_2^2$, so the population variance
$s.e._{\beta_{1,h}}^{(\mathrm{de\text{-}2S})}(v)^2
=
v' R_1 \Sigma_{UW}^{-1} \Omega_{U,h} \Sigma_{UW}^{\prime-1} R_1' v$
is well defined and bounded away from zero.

By Lemma \ref{lemma6}, under Condition \ref{cond1_de_2S},
\begin{align}\label{eq:asym_linear}
\begin{split}
\frac{\sqrt n\, v'(\hat\beta_{1,h}^{(\mathrm{de\text{-}2S})}-\beta_{1,h})}
{s.e._{\beta_{1,h}}^{(\mathrm{de\text{-}2S})}(v)}
&=
\frac{1}{\sqrt n}
\sum_{t=p}^{n-h}
\frac{v' G U_t e_{t,h}}
{s.e._{\beta_{1,h}}^{(\mathrm{de\text{-}2S})}(v)}
+ o_p(1).
\end{split}
\end{align}
Lemma \ref{lemma7} implies that the leading term converges in distribution to
$\mathcal N(0,1)$, and asymptotic normality follows by Slutsky’s theorem.

We now establish consistency of the HAC variance estimator
$R_1 (\widehat{\Sigma}_{UW}^{(\mathrm{re})})^{-1}
\hat\Omega_{U,h}^{(hac)}
(\widehat{\Sigma}_{UW}^{(\mathrm{re})})^{\prime-1} R_1'$
(the HC case is identical). Fix $v$ with $\|v\|_1=1$. Consider
\begin{align}\label{eq:hac_diff}
\begin{split}
&
\Big|
v'R_1 (\widehat{\Sigma}_{UW}^{(\mathrm{re})})^{-1}
\hat\Omega_{U,h}^{(hac)}
(\widehat{\Sigma}_{UW}^{(\mathrm{re})})^{\prime-1} R_1'v
-
v'R_1 \Sigma_{UW}^{-1}
\Omega_{U,h}
\Sigma_{UW}^{\prime-1} R_1'v
\Big|
\le
S_1 + 2S_2 + S_3 + 2S_4 + S_5 ,
\end{split}
\end{align}
where
\begin{align*}
S_1
&:=
\big|
v'R_1
\big((\widehat{\Sigma}_{UW}^{(\mathrm{re})})^{-1}-\Sigma_{UW}^{-1}\big)
\big(\hat\Omega_{U,h}^{(hac)}-\Omega_{U,h}\big)
\big((\widehat{\Sigma}_{UW}^{(\mathrm{re})})^{-1}-\Sigma_{UW}^{-1}\big)
R_1'v
\big|,\\
S_2
&:=
\big|
v'R_1
\big((\widehat{\Sigma}_{UW}^{(\mathrm{re})})^{-1}-\Sigma_{UW}^{-1}\big)
\big(\hat\Omega_{U,h}^{(hac)}-\Omega_{U,h}\big)
\Sigma_{UW}^{-1}
R_1'v
\big|,\\
S_3
&:=
\big|
v'R_1
\big((\widehat{\Sigma}_{UW}^{(\mathrm{re})})^{-1}-\Sigma_{UW}^{-1}\big)
\Omega_{U,h}
\big((\widehat{\Sigma}_{UW}^{(\mathrm{re})})^{-1}-\Sigma_{UW}^{-1}\big)
R_1'v
\big|,\\
S_4
&:=
\big|
v'R_1
\big((\widehat{\Sigma}_{UW}^{(\mathrm{re})})^{-1}-\Sigma_{UW}^{-1}\big)
\Omega_{U,h}
\Sigma_{UW}^{-1}
R_1'v
\big|,\\
S_5
&:=
\big|
v'R_1
\Sigma_{UW}^{-1}
\big(\hat\Omega_{U,h}^{(hac)}-\Omega_{U,h}\big)
\Sigma_{UW}^{-1}
R_1'v
\big|.
\end{align*}

To control inverse perturbations, note that
$(\widehat{\Sigma}_{UW}^{(\mathrm{re})})^{-1}-\Sigma_{UW}^{-1}
=
\Sigma_{UW}^{-1}
(\Sigma_{UW}-\widehat{\Sigma}_{UW}^{(\mathrm{re})})
(\widehat{\Sigma}_{UW}^{(\mathrm{re})})^{-1}$, so
\begin{align}\label{eq:inv_rate}
\begin{split}
\|(\widehat{\Sigma}_{UW}^{(\mathrm{re})})^{-1}-\Sigma_{UW}^{-1}\|_\infty
=
O_p\left(
k_W^2\|\widehat{\Sigma}_{UW}^{(\mathrm{re})}-\Sigma_{UW}\|_\infty
\right),
\end{split}
\end{align}
under Condition \ref{cond1_de_2S}, where we used
$\|\Sigma_{UW}^{-1}\|_\infty=O(k_W)$ and
$\|(\widehat{\Sigma}_{UW}^{(\mathrm{re})})^{-1}\|_\infty=O_p(k_W)$.

Since $\|v\|_1=1$ and $\|R_1\|_\infty=1$, the quadratic terms satisfy
\begin{align}\label{eq:S1S5}
\begin{split}
S_1 &\lesssim
\|(\widehat{\Sigma}_{UW}^{(\mathrm{re})})^{-1}-\Sigma_{UW}^{-1}\|_\infty^{2}
|\tilde v'(\hat\Omega_{U,h}^{(hac)}-\Omega_{U,h})\tilde v|,\\
S_2 &\lesssim
\|(\widehat{\Sigma}_{UW}^{(\mathrm{re})})^{-1}-\Sigma_{UW}^{-1}\|_\infty
|\tilde v'(\hat\Omega_{U,h}^{(hac)}-\Omega_{U,h})\tilde v|,\\
S_3 &\lesssim
\|(\widehat{\Sigma}_{UW}^{(\mathrm{re})})^{-1}-\Sigma_{UW}^{-1}\|_\infty^{2},\\
S_4 &\lesssim
\|(\widehat{\Sigma}_{UW}^{(\mathrm{re})})^{-1}-\Sigma_{UW}^{-1}\|_\infty,\\
S_5 &\lesssim
k_W^2 |\tilde v'(\hat\Omega_{U,h}^{(hac)}-\Omega_{U,h})\tilde v|.
\end{split}
\end{align}
for some $\tilde v\in \mathbb R^{dp\times 1}$, $\|\tilde v\|_1=1$.

Because
$\|(\widehat{\Sigma}_{UW}^{(\mathrm{re})})^{-1}-\Sigma_{UW}^{-1}\|_\infty=o_p(1)$,
$S_1$ is dominated by $S_2$, $S_2$ by $S_5$, and $S_3$ by $S_4$. Therefore,
\[
\Big|
\widehat{s.e.}_{\hat\beta_{1,h}^{(\mathrm{de\text{-}2S})}}^{(hac)}(v)^2
-
s.e._{\beta_{1,h}}^{(\mathrm{de\text{-}2S})}(v)^2
\Big|
\lesssim
k_W^2\|\widehat{\Sigma}_{UW}^{(\mathrm{re})}-\Sigma_{UW}\|_\infty
+
k_W^2\|\hat\Omega_{U,h}^{(hac)}-\Omega_{U,h}\|_{\max}
=
o_p(1),
\]
under Condition \ref{cond1_de_2S} and Lemma \ref{lemma8}. This establishes
consistency of the HAC variance estimator.

\end{proof}

\section{Auxiliary lemmas}\label{sec:OA_lemmas}

\begin{lemma}[Lemma A.2 of \cite{krampe2023structural}]\label{lemma1}
Let $\left\{\Phi_j^{(k)}, j=0,1, \ldots\right\}, k=1,2$, be linear filters with $\sum_{j=0}^{\infty}\big\|\Phi_j^{(k)}\big\|_2=O(1), k=1$, 2 . Then under Assumption \ref{Assump_deLS_de2S}(iv)
$$
\left\|1 / \sqrt{n} \sum_{t=1}^n \sum_{j, k=0}^{\infty} \Phi_j^{(1)}\left(u_{t-j} u_{t-k}^{\prime}-\mathbf{1}(j=k) \Sigma_{u}\right)\left(\Phi_k^{(2)}\right)^{\prime}\right\|_{\max }=O\left(\sqrt{\tilde{\nu}_n}\right)
$$
\end{lemma}

\begin{lemma}\label{lemma5}
Let 
$$
\widehat{CN}:=\frac{1}{n} \sum_{t=p}^{n-h} R_1 (\widehat{\Sigma}_{UW}^{(\mathrm{re})})^{-1} \hat U_t W_t^{\prime} R_1^{\prime}\quad \text{and} \quad CN:=\frac{1}{n} \sum_{t=p}^{n-h} R_1 \Sigma_{UW}^{-1} U_tW_t^{\prime} R_1^{\prime}.
$$
Under Assumptions \ref{Assump_deLS_de2S}(ii), and (iv)-(vi), it holds true that:
\begin{flalign*}
   (a)&\left\|\frac{1}{\sqrt{n}} \sum_{t=p}^{n - h} \Sigma_{U W}^{-1} U_t e_{t ,h}\right\|_{\text {max }}=O_p\left(\sqrt{\tilde{\nu}_n}\right) \quad\text{and} \left\|\frac{1}{n} \sum_{t=p}^{n-h} \Sigma_{UW}^{-1} U_t W_t^{\prime}-I_{d p}\right\|_{\text {max }}=O_p\left(\sqrt{\tilde{\nu}_n/n}\right) ;&\\
   (b)& \left\|\widehat{C N}^{-1}-C N^{-1}\right\|_{\text {max }}=O_p\left(\|\hat{\mathbf{A}}^{(\mathrm{re})}-\mathbf{A}\|_{\infty}k_{W}+\left\|\widehat{\Sigma}_{U W}^{(\mathrm{re})}-\Sigma_{U W}\right\|_{\infty} k_{W}\right);&\\
   (c)& \left\|\widehat{C N}^{-1}-I_p\right\|_{\text {max }}=O_p\left(\sqrt{\tilde{\nu}_n / n}+\|\hat{\mathbf{A}}^{(\mathrm{re})}-\mathbf{A}\|_{\infty}k_{W}+\left\|\widehat{\Sigma}_{U W}^{(\mathrm{re})}-\Sigma_{U W}\right\|_{\infty} k_{W}\right).&
\end{flalign*}
\end{lemma}


\begin{lemma}\label{lemma6}
If Assumption \ref{Assump_deLS_de2S} is satisfied, then for any vector $v\in\mathbb{R}^p$ such that $\|v\|_1=1$, 
\begin{equation}\label{res_lemma6}
\begin{aligned}
 &\sqrt{n} v^{\prime}\left(\hat{\beta}_{1, h}^{(de-2S)}-\beta_{1,h}\right)=v^{\prime}\left(E\left[U_{1,t}^\perp  W_{1, t}^{\prime}\right]\right)^{-1}\left(\frac{1}{\sqrt{n}} \sum_{t=p}^{n-h} U_{1, t}^\perp e_{t, h}\right) \\
&\hspace{2cm} + O_p\left(\tilde{\nu}_n/\sqrt{n}+\|\hat{\mathbf{A}}^{(\mathrm{re})}-\mathbf{A}\|_{\infty}k_{W}\sqrt{\tilde{\nu}_n }+\left\|\widehat{\Sigma}_{U W}^{(\mathrm{re})}-\Sigma_{U W}\right\|_{\infty} k_{W}\sqrt{\tilde{\nu}_n }\right.&\\
&\hspace{3cm}\left.+\left\|\hat{\beta}_{2, h}-\beta_{2, h}\right\|_{\infty}\left\{\sqrt{\tilde{\nu}_n}+\|\hat{\mathbf{A}}^{(\mathrm{re})}-\mathbf{A}\|_{\infty}k_{W}\sqrt{\tilde{\nu}_n }+\left\|\widehat{\Sigma}_{U W}^{(\mathrm{re})}-\Sigma_{U W}\right\|_{\infty} k_{W}\sqrt{\tilde{\nu}_n }\right\}\right) 
\end{aligned}
\end{equation}
\end{lemma}

\begin{lemma}\label{lemma7} If Assumptions \ref{Assump_deLS_de2S} is satisfied, then for any vector $v\in\mathbb{R}^p$ such that $\|v\|_1=1$, it holds that 
\begin{equation}\label{Asyp_clean_de-2S}
   \frac{1}{\sqrt{n}} \sum_{t=p}^{n-h}\frac{v' R_1 \Sigma_{UW}^{-1} U_t e_{t, h}}{s.e._{\hat\beta_{1,h}^{(de-2S)}}(v)} \xrightarrow{d} \mathcal{N}(0,1).
\end{equation}

\end{lemma}

\begin{lemma}\label{lemma8}
Let $\{W_t\}$ follow a VAR$(p)$ model and suppose that
Assumption~\ref{Assump_deLS_de2S} hold.
Then, for any given vector $v\in\mathbb R^{dp}$ satisfying
$\|v\|_1=1$ and $k_A^{4.5}\big(\nu_n/n\big)^{(1-\mu)/2}=o_p(1)$,
\begin{align}
\label{s_t}
\begin{split}
    &|v' (\Omega_{U,h}^{(hac)}-\Omega_{U,h})v | \xrightarrow{p} 0,\\
    &|v' (\Omega_{U,h}^{(hc)}-\Omega_{U,h})v | \xrightarrow{p} 0,
\end{split} 
\end{align}

\end{lemma}

\section{Proof of lemmas}\label{sec:OA_cltvar}

\begin{proof}[\textbf{Proof of Lemma \ref{lemma1}}] 
See Appendix A of \cite{krampe2023structural}.
\end{proof}

\begin{proof}[\textbf{Proof of Lemma \ref{lemma5}}]
Recall the $\operatorname{VAR}$ representation: under stability,
$W_t=\sum_{k=0}^{\infty}\Psi_k u_{t-k}$.
Let $\tilde e_{pj}$, $j=1,\ldots,p$, denote the $p$-dimensional unit vectors.
Then
\begin{align}
\label{eq:Ut_def}
U_t
:=
\big(u_t',u_{t-1}',\ldots,u_{t-p+1}'\big)'=
\sum_{j=0}^{p-1}\big(\tilde e_{p(j+1)}\otimes I_d\big)u_{t-j},
\end{align}
and therefore $\Sigma_{UW}
:=
\mathbb E[U_tW_t']
=
\sum_{j=0}^{p-1}\big(\tilde e_{p(j+1)}\otimes I_d\big)\Sigma_u\Psi_j'.$

\textbf{Part (a).}
For the second assertion, define the matrix-valued filters
\begin{align}
\Phi_j^{(1)}=
\begin{cases}
\Sigma_{UW}^{-1}\big(\tilde e_{p(j+1)}\otimes I_d\big), & j=0,\ldots,p-1,\\
0,& j\ge p,
\end{cases}
\qquad
\Phi_k^{(2)}=\Psi_k,\; k\ge 0.
\label{eq:filters_a2}
\end{align}
Then $\Sigma_{UW}^{-1}U_t=\sum_{j\ge 0}\Phi_j^{(1)}u_{t-j}$ and
$W_t=\sum_{k\ge 0}\Phi_k^{(2)}u_{t-k}$.
Moreover,
$\sum_{j\ge 0}\|\Phi_j^{(1)}\|_2 \le p\|\Sigma_{UW}^{-1}\|_2=O(1)$
(and similarly for $\|\cdot\|_\infty$ if needed),
while $\sum_{k\ge 0}\|\Phi_k^{(2)}\|_2=\sum_{k\ge 0}\|\Psi_k\|_2=O(1)$ by stability.
Applying Lemma \ref{lemma1} to the pair of linear processes yields
\begin{align}
\Bigg\|
\frac{1}{n}\sum_{t=p}^{n-h}\Sigma_{UW}^{-1}U_tW_t'
-
\mathbb E\big[\Sigma_{UW}^{-1}U_tW_t'\big]
\Bigg\|_{\max}
=
O_p\Big(\sqrt{\tilde\nu_n/n}\Big).
\label{eq:a_second_rate}
\end{align}
Since $\mathbb E[\Sigma_{UW}^{-1}U_tW_t']=\Sigma_{UW}^{-1}\Sigma_{UW}=I_{dp}$,
this is exactly the second result in part (a).

For the first assertion, note that
\begin{align}
\frac{1}{\sqrt n}\sum_{t=p}^{n-h}\Sigma_{UW}^{-1}U_t e_{t,h}
=
\frac{1}{\sqrt n}\sum_{t=p+h}^{n}\Sigma_{UW}^{-1}U_{t-h}\, e_{t-h,h}.
\label{eq:shift_index}
\end{align}
Using \eqref{eq:Ut_def},
\begin{align}
U_{t-h}
=
\sum_{j=h}^{p+h-1}\big(\tilde e_{p(j-h+1)}\otimes I_d\big)u_{t-j},
\qquad
e_{t-h,h}
=
\sum_{k=0}^{h-1} u_{t-k}'\, J(\mathbf A')^k J' v_1.
\label{eq:U_and_e_linear}
\end{align}
Equivalently, $e_{t-h,h}=\sum_{k=0}^{h-1}\big(v_1'J\mathbf A^kJ'\big)u_{t-k}$.
Define the filters
\begin{align}
\Phi_j^{(1)}=
\begin{cases}
\Sigma_{UW}^{-1}\big(\tilde e_{p(j-h+1)}\otimes I_d\big), & j=h,\ldots,p+h-1,\\
0, & \text{otherwise},
\end{cases}
\qquad
\Phi_k^{(2)}=
\begin{cases}
v_1'J\mathbf A^kJ', & k=0,\ldots,h-1,\\
0,& k\ge h.
\end{cases}
\label{eq:filters_a1}
\end{align}
Then $\Sigma_{UW}^{-1}U_{t-h}=\sum_{j\ge 0}\Phi_j^{(1)}u_{t-j}$ and
$e_{t-h,h}=\sum_{k\ge 0}\Phi_k^{(2)}u_{t-k}$, with
$\sum_j\|\Phi_j^{(1)}\|_2=O(1)$ and $\sum_k\|\Phi_k^{(2)}\|_2=O(1)$.
Applying Lemma \ref{lemma1} again gives the first result in part (a).

\textbf{Part (b).}
Write
\begin{align}
\label{eq:CN_decomp}
\begin{split}
\|\widehat{CN}-CN\|_{\max}
&=
\Bigg\|
\frac{1}{n}\sum_{t=p}^{n-h}
R_1\Big((\widehat{\Sigma}_{UW}^{(\mathrm{re})})^{-1}\hat U_t-\Sigma_{UW}^{-1}U_t\Big)
W_t'R_1'
\Bigg\|_{\max}\\
&\le \tilde I_1+\tilde I_2+\tilde I_3,
\end{split}
\end{align}
where
\begin{align}
\tilde I_1
&:=
\Bigg\|
\frac{1}{n}\sum_{t=p}^{n-h}
R_1\big((\widehat{\Sigma}_{UW}^{(\mathrm{re})})^{-1}-\Sigma_{UW}^{-1}\big)(\hat U_t-U_t)
W_t'R_1'
\Bigg\|_{\max},\notag\\
\tilde I_2
&:=
\Bigg\|
\frac{1}{n}\sum_{t=p}^{n-h}
R_1\big((\widehat{\Sigma}_{UW}^{(\mathrm{re})})^{-1}-\Sigma_{UW}^{-1}\big)U_t
W_t'R_1'
\Bigg\|_{\max},\notag\\
\tilde I_3
&:=
\Bigg\|
\frac{1}{n}\sum_{t=p}^{n-h}
R_1\Sigma_{UW}^{-1}(\hat U_t-U_t)W_t'R_1'
\Bigg\|_{\max}.\notag
\end{align}

First, Lemma \ref{lemma1} (applied to suitable filters) implies that for each fixed $j$,
\begin{align}
\Big\|
\frac{1}{n}\sum_{t} W_{t+j}W_t' - \Sigma_W(j)
\Big\|_{\max}
=
O_p\Big(\sqrt{\tilde\nu_n/n}\Big),
\qquad
\Sigma_W(j):=\mathbb E[W_{t+j}W_t'].
\label{eq:W_W_rate}
\end{align}
Since $\|\Sigma_W(j)\|_{\max}=O(1)$, we have
$\big\|\frac{1}{n}\sum_t W_{t+j}W_t'\big\|_{\max}=O_p(1)$.

Next, using $\hat u_t-u_t=J(\mathbf A-\widehat{\mathbf A}^{(\mathrm{re})})W_{t-1}$ and \eqref{eq:Ut_def},
\begin{align}
\hat U_t-U_t
&=
\sum_{j=0}^{p-1}\big(\tilde e_{p(j+1)}\otimes I_d\big)(\hat u_{t-j}-u_{t-j}) =
\sum_{j=0}^{p-1}\big(\tilde e_{p(j+1)}\otimes I_d\big)J(\mathbf A-\widehat{\mathbf A}^{(\mathrm{re})})W_{t-j-1}.
\label{eq:Uhat_minus_U}
\end{align}
Using $\|J\|_\infty=1$ and $\|\tilde e_{p(j+1)}\otimes I_d\|_\infty=1$,
\begin{align}
\label{eq:UW_cross_bound}
\begin{split}
\Big\|
\frac{1}{n}\sum_{t=p}^{n-h}(\hat U_t-U_t)W_t'
\Big\|_{\max}
&\le
\sum_{j=0}^{p-1}\|\widehat{\mathbf A}^{(\mathrm{re})}-\mathbf A\|_\infty\,
\Big\|
\frac{1}{n}\sum_{t=p}^{n-h}W_{t-j-1}W_t'
\Big\|_{\max}=
O_p\big(\|\widehat{\mathbf A}^{(\mathrm{re})}-\mathbf A\|_\infty\big).
\end{split}
\end{align}

Also, the second assertion in part (a) yields
\begin{align}
\Big\|
\frac{1}{n}\sum_{t=p}^{n-h}\Sigma_{UW}^{-1}U_tW_t'
\Big\|_{\max}
\le
\Big\|
\frac{1}{n}\sum_{t=p}^{n-h}\Sigma_{UW}^{-1}U_tW_t' - I_{dp}
\Big\|_{\max}
+\|I_{dp}\|_{\max} =
O_p(1).
\label{eq:SigmaInvUW_UW_avg}
\end{align}

Now use $\|R_1\|_\infty=1$, $\|\Sigma_{UW}^{-1}\|_\infty=O(k_W)$,
$\|(\widehat{\Sigma}_{UW}^{(\mathrm{re})})^{-1}\|_\infty=O_p(k_W)$, and
\begin{align}
\|(\widehat{\Sigma}_{UW}^{(\mathrm{re})})^{-1}-\Sigma_{UW}^{-1}\|_\infty
\le
\|(\widehat{\Sigma}_{UW}^{(\mathrm{re})})^{-1}\|_\infty
\|\widehat{\Sigma}_{UW}^{(\mathrm{re})}-\Sigma_{UW}\|_\infty
\|\Sigma_{UW}^{-1}\|_\infty.
\label{eq:inv_perturb}
\end{align}
Combining \eqref{eq:UW_cross_bound}--\eqref{eq:inv_perturb} gives
\begin{align}
\tilde I_1
&\le
\|(\widehat{\Sigma}_{UW}^{(\mathrm{re})})^{-1}-\Sigma_{UW}^{-1}\|_\infty\,
\Big\|
\frac{1}{n}\sum_{t=p}^{n-h}(\hat U_t-U_t)W_t'
\Big\|_{\max}
=
O_p\Big(\|\widehat{\Sigma}_{UW}^{(\mathrm{re})}-\Sigma_{UW}\|_\infty\,
\|\widehat{\mathbf A}^{(\mathrm{re})}-\mathbf A\|_\infty\, k_W^2\Big),\notag\\
\tilde I_2
&\le
\|(\widehat{\Sigma}_{UW}^{(\mathrm{re})})^{-1}\|_\infty\,
\|\widehat{\Sigma}_{UW}^{(\mathrm{re})}-\Sigma_{UW}\|_\infty\,
\Big\|
\frac{1}{n}\sum_{t=p}^{n-h}\Sigma_{UW}^{-1}U_tW_t'
\Big\|_{\max}
=
O_p\Big(\|\widehat{\Sigma}_{UW}^{(\mathrm{re})}-\Sigma_{UW}\|_\infty\, k_W\Big),\notag\\
\tilde I_3
&\le
\|\Sigma_{UW}^{-1}\|_\infty\,
\Big\|
\frac{1}{n}\sum_{t=p}^{n-h}(\hat U_t-U_t)W_t'
\Big\|_{\max}
=
O_p\Big(\|\widehat{\mathbf A}^{(\mathrm{re})}-\mathbf A\|_\infty\, k_W\Big).\label{eq:I123_rates}
\end{align}
Since $\tilde I_1$ is of higher order (a product of two estimation errors), \eqref{eq:I123_rates} implies
\begin{align}
\|\widehat{CN}-CN\|_{\max}
=
O_p\Big(
\|\widehat{\mathbf A}^{(\mathrm{re})}-\mathbf A\|_\infty\, k_W
+
\|\widehat{\Sigma}_{UW}^{(\mathrm{re})}-\Sigma_{UW}\|_\infty\, k_W
\Big).
\label{eq:CN_rate}
\end{align}

For the corresponding inverse, use
$\widehat{CN}^{-1}-CN^{-1}=\widehat{CN}^{-1}(CN-\widehat{CN})CN^{-1}$ and
$\|AXB\|_{\max}\le \|A\|_\infty\|X\|_{\max}\|B\|_1$ to obtain
\begin{align}
\|\widehat{CN}^{-1}-CN^{-1}\|_{\max}
\le
\|\widehat{CN}^{-1}\|_\infty\,
\|\widehat{CN}-CN\|_{\max}\,
\|CN^{-1}\|_1,
\label{eq:CNinv_bound}
\end{align}
which yields the stated rate under $\|\widehat{CN}^{-1}\|_\infty=O_p(1)$ and $\|CN^{-1}\|_1=O(1)$.

\textbf{Part (c).}
By part (a) and $R_1R_1'=I_p$,
\begin{align}
\|CN-I_p\|_{\max}
=
O_p\Big(\sqrt{\tilde\nu_n/n}\Big).
\label{eq:CN_minus_I}
\end{align}
Hence, by the triangle inequality and \eqref{eq:CN_rate},
\begin{align}
\|\widehat{CN}-I_p\|_{\max}
\le
\|\widehat{CN}-CN\|_{\max}+\|CN-I_p\|_{\max}
=
O_p\Big(
\sqrt{\tilde\nu_n/n}
+
\|\widehat{\mathbf A}^{(\mathrm{re})}-\mathbf A\|_\infty\, k_W
+
\|\widehat{\Sigma}_{UW}^{(\mathrm{re})}-\Sigma_{UW}\|_\infty\, k_W
\Big).
\label{eq:CNhat_minus_I}
\end{align}
Finally, $\widehat{CN}^{-1}-I_p=\widehat{CN}^{-1}(I_p-\widehat{CN})$, so
$\|\widehat{CN}^{-1}-I_p\|_{\max}\le \|\widehat{CN}^{-1}\|_\infty\|\widehat{CN}-I_p\|_{\max}$,
which completes the proof.
\end{proof}

\begin{proof}[\textbf{Proof of Lemma \ref{lemma6}}]

By the definition of the de-2S estimator,
\begin{align}\label{eq1_lemma6}
\begin{split}
\sqrt{n}\,
v^{\prime}\!\left(\hat{\beta}_{1, h}^{(\mathrm{de\text{-}2S})}-\beta_{1,h}\right)
&=
v^{\prime}\!\left(\frac{1}{n}\sum_{t=p}^{n-h}\hat U_{1,t}^{\perp}W_{1,t}^{\prime}\right)^{-1}
\!\left(\frac{1}{\sqrt{n}}\sum_{t=p}^{n-h}\hat U_{1,t}^{\perp}e_{t,h}\right)\\
&\quad+
v^{\prime}\!\left(\frac{1}{n}\sum_{t=p}^{n-h}\hat U_{1,t}^{\perp}W_{1,t}^{\prime}\right)^{-1}
\!\left(\frac{1}{\sqrt{n}}\sum_{t=p}^{n-h}\hat U_{1,t}^{\perp}
W_{2,t}^{\prime}(\beta_{2,h}-\hat\beta_{2,h})\right)\\
&=
v^{\prime}\!\left(E[U_{1,t}^{\perp}W_{1,t}^{\prime}]\right)^{-1}
\!\left(\frac{1}{\sqrt{n}}\sum_{t=p}^{n-h}U_{1,t}^{\perp}e_{t,h}\right)
+\tilde\Lambda_0+\tilde\Lambda_1+\tilde\Lambda_2 .
\end{split}
\end{align}

where
\begin{align}
\begin{split}
\tilde{\Lambda}_0
&:=
v^{\prime}\!\left\{
\left(\frac{1}{n}\sum_{t=p}^{n-h}U_{1,t}^{\perp}W_{1,t}^{\prime}\right)^{-1}
-
\left(E[U_{1,t}^{\perp}W_{1,t}^{\prime}]\right)^{-1}
\right\}
\left(\frac{1}{\sqrt{n}}\sum_{t=p}^{n-h}U_{1,t}^{\perp}e_{t,h}\right),
\label{eq4_lemma6}
\\
\tilde{\Lambda}_1
&:=
v^{\prime}\!\left(\frac{1}{n}\sum_{t=p}^{n-h}\hat U_{1,t}^{\perp}W_{1,t}^{\prime}\right)^{-1}
\!\left(\frac{1}{\sqrt{n}}\sum_{t=p}^{n-h}\hat U_{1,t}^{\perp}e_{t,h}\right)
-
v^{\prime}\!\left(\frac{1}{n}\sum_{t=p}^{n-h}U_{1,t}^{\perp}W_{1,t}^{\prime}\right)^{-1}
\!\left(\frac{1}{\sqrt{n}}\sum_{t=p}^{n-h}U_{1,t}^{\perp}e_{t,h}\right),
\\
\tilde{\Lambda}_2
&:=
v^{\prime}\!\left(\frac{1}{n}\sum_{t=p}^{n-h}\hat U_{1,t}^{\perp}W_{1,t}^{\prime}\right)^{-1}
\!\left(\frac{1}{\sqrt{n}}\sum_{t=p}^{n-h}\hat U_{1,t}^{\perp}
W_{2,t}^{\prime}(\beta_{2,h}-\hat\beta_{2,h})\right).
\end{split}
\end{align}

Using
\[
U_{1,t}^{\perp}
=
\left(R_1\Sigma_{UW}^{-1}R_1^{\prime}\right)^{-1}R_1\Sigma_{UW}^{-1}U_t,
\qquad
\hat U_{1,t}^{\perp}
=
\left(R_1(\widehat{\Sigma}_{UW}^{(\mathrm{re})})^{-1}R_1^{\prime}\right)^{-1}
R_1(\widehat{\Sigma}_{UW}^{(\mathrm{re})})^{-1}U_t,
\]
$\tilde{\Lambda}_1$ can be rewritten as
\begin{align}
\begin{split}
\tilde{\Lambda}_1
&=
\frac{1}{\sqrt{n}}\sum_{t=p}^{n-h}
v^{\prime}\!\left(
\widehat{C N}^{-1}R_1(\widehat{\Sigma}_{UW}^{(\mathrm{re})})^{-1}\hat U_t
-
C N^{-1}R_1\Sigma_{UW}^{-1}U_t
\right)e_{t,h} \\
&=
\tilde{\Lambda}_{11}+\tilde{\Lambda}_{12}+\tilde{\Lambda}_{13}+\tilde{\Lambda}_{14}.
\end{split}
\end{align}

where $\widehat{C N}$ and $C N$ are defined as in Lemma \ref{lemma5}, and
\begin{align}
\begin{split}
\tilde{\Lambda}_{11}
&:=
\frac{1}{\sqrt{n}}\sum_{t=p}^{n-h}
v^{\prime}(\widehat{C N}^{-1}-C N^{-1})
R_1\big((\widehat{\Sigma}_{UW}^{(\mathrm{re})})^{-1}-\Sigma_{UW}^{-1}\big)
\hat U_t e_{t,h},
\\
\tilde{\Lambda}_{12}
&:=
\frac{1}{\sqrt{n}}\sum_{t=p}^{n-h}
v^{\prime}(\widehat{C N}^{-1}-C N^{-1})
R_1\Sigma_{UW}^{-1}\hat U_t e_{t,h},
\label{eq2_lemma6}
\\
\tilde{\Lambda}_{13}
&:=
\frac{1}{\sqrt{n}}\sum_{t=p}^{n-h}
v^{\prime}C N^{-1}R_1
\big((\widehat{\Sigma}_{UW}^{(\mathrm{re})})^{-1}-\Sigma_{UW}^{-1}\big)
\hat U_t e_{t,h},
\\
\tilde{\Lambda}_{14}
&:=
\frac{1}{\sqrt{n}}\sum_{t=p}^{n-h}
v^{\prime}C N^{-1}R_1\Sigma_{UW}^{-1}(\hat U_t-U_t)e_{t,h}.
\end{split}
\end{align}

Similarly,
\begin{align}
\begin{split}
\tilde{\Lambda}_2
&=
\frac{1}{\sqrt{n}}\sum_{t=p}^{n-h}
v^{\prime}\widehat{C N}^{-1}R_1(\widehat{\Sigma}_{UW}^{(\mathrm{re})})^{-1}
\hat U_t W_t^{\prime}R_2^{\prime}(\beta_{2,h}-\hat\beta_{2,h})\\
&=
\big(\tilde{\Lambda}_{21}+\tilde{\Lambda}_{22}+\tilde{\Lambda}_{23}+\tilde{\Lambda}_{24}\big)
(\beta_{2,h}-\hat\beta_{2,h}),
\end{split}
\end{align}

where
\begin{align}
\begin{split}
\tilde{\Lambda}_{21}
&:=
\frac{1}{\sqrt{n}}\sum_{t=p}^{n-h}
v^{\prime}R_1\Sigma_{UW}^{-1}\hat U_t W_t^{\prime}R_2^{\prime},
\\
\tilde{\Lambda}_{22}
&:=
\frac{1}{\sqrt{n}}\sum_{t=p}^{n-h}
v^{\prime}R_1\big((\widehat{\Sigma}_{UW}^{(\mathrm{re})})^{-1}-\Sigma_{UW}^{-1}\big)
\hat U_t W_t^{\prime}R_2^{\prime},
\\
\tilde{\Lambda}_{23}
&:=
\frac{1}{\sqrt{n}}\sum_{t=p}^{n-h}
v^{\prime}(\widehat{C N}^{-1}-I_p)
R_1\Sigma_{UW}^{-1}\hat U_t W_t^{\prime}R_2^{\prime},
\label{eq3_lemma6}
\\
\tilde{\Lambda}_{24}
&:=
\frac{1}{\sqrt{n}}\sum_{t=p}^{n-h}
v^{\prime}(\widehat{C N}^{-1}-I_p)
R_1\big((\widehat{\Sigma}_{UW}^{(\mathrm{re})})^{-1}-\Sigma_{UW}^{-1}\big)
\hat U_t W_t^{\prime}R_2^{\prime}.
\end{split}
\end{align}
Although $\tilde{\Lambda}_2$ is written in terms of the estimated innovations
$\hat U_t$, the difference between $\hat U_t$ and $U_t$ is asymptotically
negligible under Condition~1 in the expansion of
$\tilde{\Lambda}_2(\beta_{2,h}-\hat{\beta}_{2,h})$. The remaining arguments proceed by establishing the stochastic orders of the
terms in \eqref{eq4_lemma6}, \eqref{eq2_lemma6}, and \eqref{eq3_lemma6}.  
The conclusion then follows by substituting the derived rates into
\eqref{eq1_lemma6} and noting that higher-order terms are asymptotically negligible.
\end{proof}

\begin{proof}[\textbf{Proof of Lemma \ref{lemma7}}]
The proof applies Theorem~5.20 of \cite{white2000asymptotic} to the (possibly
triangular) array
\[
z_{nt}
:=
v^{\prime} R_1 \Sigma_{UW}^{-1} U_t \, e_{t,h},
\qquad t=p,\ldots,n-h,
\]
where the dependence on $n$ is through the dimension $d$.

\textit{Step 1: Mean zero.}
By definition, $e_{t,h}$ is the population linear projection error at horizon $h$,
so $\mathbb E(e_{t,h}\mid \mathcal F_t)=0$. Since $U_t$ is $\mathcal F_t$-measurable, $\mathbb E[z_{nt}]=0 .$

\textit{Step 2: Mixing.}
Assumption~\ref{Assump_deLS_de2S}(\textit{v}) implies that $\{u_t\}$ is strongly
mixing with coefficients $\alpha(\cdot)$ of size $-r/(r-2)$ for some $r>2$.
Because $U_t$ is a measurable function of $(u_t,\ldots,u_{t-p+1})$ and $e_{t,h}$ is
a measurable function of $(u_{t+1},\ldots,u_{t+h})$ with fixed $(p,h)$, the process
$\{(U_t,e_{t,h})\}$ is strongly mixing with the same mixing size. Hence
$\{z_{nt}\}$ is strongly mixing of size $-r/(r-2)$ as a measurable transformation
of $(U_t,e_{t,h})$ (see, e.g., Proposition~3.50 of \cite{white2000asymptotic}).

\textit{Step 3: Uniform $r$th moment bound.}
Using $U_t=\sum_{j=0}^{p-1}(\tilde e_{p(j+1)}\otimes I_d)u_{t-j}$ and the linear
representation $e_{t,h}=\sum_{k=0}^{h-1} v_1^{\prime}\Psi_k u_{t+h-k}$ (with fixed
$p$ and $h$), we can write
\[
z_{nt}
=
\sum_{j=0}^{p-1}\sum_{k=0}^{h-1}
\bigl(a_{j}^{\prime}u_{t-j}\bigr)\bigl(b_k^{\prime}u_{t+h-k}\bigr),
\quad
a_{j}:=\left(\tilde e_{p(j+1)}^{\prime}\otimes I_d\right)\Sigma_{UW}^{-1}R_1^{\prime}v,
\quad
b_k:=\Psi_k^{\prime}v_1 .
\]
By $(\sum_{\ell=1}^m |x_\ell|)^r \le m^{r-1}\sum_{\ell=1}^m |x_\ell|^r$ and H\"older's
inequality,
\[
\mathbb E|z_{nt}|^r
\le
(ph)^{r-1}\sum_{j=0}^{p-1}\sum_{k=0}^{h-1}
\Bigl(\mathbb E|a_{j}^{\prime}u_{t-j}|^{2r}\Bigr)^{1/2}
\Bigl(\mathbb E|b_{k}^{\prime}u_{t+h-k}|^{2r}\Bigr)^{1/2}\lesssim \|a_j\|_1^r \|b_k\|_1^r .
\]
due to Assumption~\ref{Assump_deLS_de2S}(\textit{v}). Therefore,
 $\sup_{n}\sup_{t}\mathbb E|z_{nt}|^r < \infty .$

\textit{Step 4: Nondegeneracy of the asymptotic variance.} 
By the definition of $\Omega_{U,h}$ as the long-run variance of the score process
$R_1\Sigma_{UW}^{-1}U_t e_{t,h}$, since $\Omega_{U,h}$ is positive definite with $\lambda_{\min}(\Omega_{U,h})\ge C^{-1}$ and $\Sigma_{UW}$ is positive definite by its construction, it follows that the asymptotic variance is nondegenerate.

Steps 1--4 verify the conditions of Theorem~5.20 in \cite{white2000asymptotic}. The proof is therefore complete.
\end{proof}

\begin{proof}[\textbf{Proof of Lemma \ref{lemma8}}]
Fix any
$v\in\mathbb R^{dp}$ with $\|v\|_1=1$. We prove the HC convergence in detail; the
HAC case follows from the same autocovariance consistency plus a standard kernel
argument.

\medskip
\noindent\textit{Step 1 (clean term).}
Let $z_t:=v's_t$. Since $v'\Omega_{U,h}v=\mathbb E(z_t^2)$, it suffices to show
$n^{-1}\sum_{t=p}^{n-h}(z_t^2-\mathbb Ez_t^2)=o_p(1)$. Recall
$s_t:= (e_{t,h},e_{t+1,h},\ldots,e_{t+p-1,h})\otimes u_t$, where each $e_{t+i,h}$
is a finite linear combination of $\{u_{t+i+1},\ldots,u_{t+i+h}\}$. Because $p$
and $h$ are fixed, $s_t$ is a measurable function of a finite block of $\{u_t\}$.
By Assumption \ref{Assump_deLS_de2S}(\textit{v}), $\{u_t\}$ is $\alpha$-mixing with
size $-r/(r-2)$ for some $r>2$, and this property is preserved under finite
blocking and measurable transformations. Hence $\{z_t^2\}$ is also $\alpha$-mixing
with the same size, which satisfies the mixing-size requirement in
Corollary~3.48 of \cite{white2000asymptotic}.

We verify the moment condition. It is enough to show $\mathbb E|z_t|^{2(r+\delta)}<\infty$
for some $\delta>0$. Since $p$ is fixed, it suffices to control a generic product
$\lambda'u_t\,e_{t,h}$ with $\|\lambda\|_2=1$. Using the representation
$e_{t,h}=\sum_{i=0}^{h-1}v_1'J\mathbf A^iJ'u_{t+h-i}$ and the inequality
$|\sum_{i=0}^{h-1}x_i|^{q}\le C\sum_{i=0}^{h-1}|x_i|^{q}$ for $q\ge1$, we obtain
\begin{align}
\label{eq:L8_mom1}
\begin{split}
\mathbb E|\lambda'u_t\,e_{t,h}|^{2r+2\delta}
&\le
C\sum_{i=0}^{h-1}\mathbb E\big|\lambda'u_t\cdot u_{t+h-i}'J\mathbf A^iJ'v_1\big|^{2r+2\delta}  \\
&\le
C\sum_{i=0}^{h-1}
\Big(\mathbb E|\lambda'u_t|^{4r+4\delta}\Big)^{1/2}
\Big(\mathbb E|u_{t}'J\mathbf A^iJ'v_1|^{4r+4\delta}\Big)^{1/2}.
\end{split}
\end{align}
Moreover, since $\dim(v_1)$ is fixed and $\|J\|_1=1$,
\[
|u_t'J\mathbf A^iJ'v_1|
\le
\|J\mathbf A^iJ'v_1\|_1\sup_{\|a\|_1=1}|a'u_t|
\le
\|\mathbf A^i\|_1\|v_1\|_1\sup_{\|a\|_1=1}|a'u_t|.
\]
By Assumption \ref{Assump_deLS_de2S}(\textit{v}), the directional moments
$\sup_{\|a\|_1=1}\mathbb E|a'u_t|^{4r+4\delta}$ are uniformly bounded for some
$\delta>0$, and by Assumption \ref{Assump_deLS_de2S}(\textit{ii}),
$\|\mathbf A^i\|_1\le Ck_A\varphi^i$. Combining these bounds in
\eqref{eq:L8_mom1} yields $\mathbb E|z_t|^{2(r+\delta)}<\infty$. Therefore
Corollary~3.48 of \cite{white2000asymptotic} applies to $\{z_t^2\}$ and gives
$n^{-1}\sum_{t=p}^{n-h}(z_t^2-\mathbb Ez_t^2)=o_p(1)$, i.e.
$n^{-1}\sum_{t=p}^{n-h}(v's_t)^2\to_p v'\Omega_{U,h}v$.

\medskip
\noindent\textit{Step 2 (estimation bias).}
Let $\hat z_t:=v'\hat s_t$. We show
$n^{-1}\sum_{t=p}^{n-h}(\hat z_t^{\,2}-z_t^2)=o_p(1)$. Since $p$ is fixed and
$\hat s_t=((\hat e_{t,h},\ldots,\hat e_{t+p-1,h})\otimes\hat u_t)$, without loss of generosity, we prove a representative block with $a=b=0$, namely
\begin{align}
\label{represen:block}
    \Bigg\|
\frac{1}{n}\sum_{t=p}^{n-h}\hat u_t\hat u_t'\hat e_{t,h}^{\,2}
-
\frac{1}{n}\sum_{t=p}^{n-h}u_tu_t'e_{t,h}^{\,2}
\Bigg\|_{\max}=o_p(1).
\end{align}

Write $\Delta u_t=\hat u_t-u_t$ and $\Delta e_{t,h}=\hat e_{t,h}-e_{t,h}$. Then
\begin{align}
\label{eq:L8_decomp}
\begin{split}
\hat u_t\hat u_t'\hat e_{t,h}^{\,2}-u_tu_t'e_{t,h}^{\,2}
&=
u_tu_t'(\hat e_{t,h}^{\,2}-e_{t,h}^{\,2})
+(\hat u_t\hat u_t'-u_tu_t')e_{t,h}^{\,2}
+(\hat u_t\hat u_t'-u_tu_t')(\hat e_{t,h}^{\,2}-e_{t,h}^{\,2}).
\end{split}
\end{align}
Taking $n^{-1}\sum_{t=p}^{n-h}(\cdot)$ and $\|\cdot\|_{\max}$, and applying
elementwise Cauchy--Schwarz to each term in \eqref{eq:L8_decomp}, we obtain
\begin{align}
\label{eq:L8_CS}
\begin{split}
\Bigg\|
\frac{1}{n}\sum_{t=p}^{n-h}\hat u_t\hat u_t'\hat e_{t,h}^{\,2}
-
\frac{1}{n}\sum_{t=p}^{n-h}u_tu_t'e_{t,h}^{\,2}
\Bigg\|_{\max}
\;\le\;
&
\max_{i,j}
\Bigg(
\frac{1}{n}\sum_{t=p}^{n-h} u_{i,t}^2u_{j,t}^2
\Bigg)^{1/2}
\Bigg(
\frac{1}{n}\sum_{t=p}^{n-h}(\hat e_{t,h}^{\,2}-e_{t,h}^{\,2})^2
\Bigg)^{1/2} \\
&+
\max_{i,j}
\Bigg(
\frac{1}{n}\sum_{t=p}^{n-h}(\hat u_{i,t}\hat u_{j,t}-u_{i,t}u_{j,t})^2
\Bigg)^{1/2}
\Bigg(
\frac{1}{n}\sum_{t=p}^{n-h} e_{t,h}^4
\Bigg)^{1/2} \\
&+
\max_{i,j}
\Bigg(
\frac{1}{n}\sum_{t=p}^{n-h}(\hat u_{i,t}\hat u_{j,t}-u_{i,t}u_{j,t})^2
\Bigg)^{1/2}
\Bigg(
\frac{1}{n}\sum_{t=p}^{n-h}(\hat e_{t,h}^{\,2}-e_{t,h}^{\,2})^2
\Bigg)^{1/2}.
\end{split}
\end{align}

We bound the factors. First,
\[
\max_{i,j}\frac{1}{n}\sum_{t=p}^{n-h} u_{i,t}^2u_{j,t}^2
\le
\max_{i}\frac{1}{n}\sum_{t=p}^{n-h} u_{i,t}^4
=
O_p(1)
\]
by Assumption~\ref{Assump_deLS_de2S}(\textit{v}). Second, $e_{t,h}$ is a finite
linear process with $\ell_1$ coefficients of order $O(k_A)$, hence
\[
\frac{1}{n}\sum_{t=p}^{n-h} e_{t,h}^4
=
O_p(k_A^4).
\]
Third, using
\[
\hat u_{i,t}\hat u_{j,t}-u_{i,t}u_{j,t}
=
\Delta u_{i,t}u_{j,t}+u_{i,t}\Delta u_{j,t}+\Delta u_{i,t}\Delta u_{j,t}
\]
and $(a+b+c)^2\le 3(a^2+b^2+c^2)$, another Cauchy--Schwarz step gives
\begin{align}
\label{eq:L8_du}
\begin{split}
\max_{i,j}\frac{1}{n}\sum_{t=p}^{n-h}(\hat u_{i,t}\hat u_{j,t}-u_{i,t}u_{j,t})^2
\;\lesssim\;
\Bigg(\max_i\frac{1}{n}\sum_{t=p}^{n-h} \Delta u_{i,t}^4\Bigg)^{1/2}
\Bigg(\max_i\frac{1}{n}\sum_{t=p}^{n-h} u_{i,t}^4\Bigg)^{1/2}
+
\max_i\frac{1}{n}\sum_{t=p}^{n-h} \Delta u_{i,t}^4.
\end{split}
\end{align}
Since $\Delta u_t=(\widehat{\mathbf A}^{(\mathrm{re})}-\mathbf A)W_{t-1}$, we have
$|\Delta u_{i,t}|
\le \|\widehat{\mathbf A}^{(\mathrm{re})}-\mathbf A\|_\infty\|W_{t-1}\|_{\max}$
and thus
\[
\max_i\frac{1}{n}\sum_{t=p}^{n-h} \Delta u_{i,t}^4
=
O_p\!\left(\|\widehat{\mathbf A}^{(\mathrm{re})}-\mathbf A\|_\infty^4\right)
\]
under Assumption~\ref{Assump_deLS_de2S}(\textit{v}). Plugging into
\eqref{eq:L8_du} yields
\[
\max_{i,j}
\Bigg(\frac{1}{n}\sum_{t=p}^{n-h}(\hat u_{i,t}\hat u_{j,t}-u_{i,t}u_{j,t})^2\Bigg)^{1/2}
=
O_p\!\left(\|\widehat{\mathbf A}^{(\mathrm{re})}-\mathbf A\|_\infty\right).
\]

Finally, since
$\Delta e_{t,h}=v_1'\big((\widehat{\mathbf A}^{(\mathrm{re})})^h-\mathbf A^h\big)W_t$,
\[
(\hat e_{t,h}^{\,2}-e_{t,h}^{\,2})^2
=
(2e_{t,h}\Delta e_{t,h}+\Delta e_{t,h}^2)^2
\le
8e_{t,h}^2\Delta e_{t,h}^2+2\Delta e_{t,h}^4,
\]
hence
\[
\frac{1}{n}\sum_{t=p}^{n-h}(\hat e_{t,h}^{\,2}-e_{t,h}^{\,2})^2
=
O_p\!\left(
k_A^2\big\|(\widehat{\mathbf A}^{(\mathrm{re})})^h-\mathbf A^h\big\|_\infty^2
\right),
\]
using $\frac{1}{n}\sum_{t=p}^{n-h} e_{t,h}^4=O_p(k_A^4)$ and
$\max_i \frac{1}{n}\sum_{t=p}^{n-h}\|W_{i,t}\|_\infty^4=O_p(1)$.

Substituting these bounds into \eqref{eq:L8_CS} gives
\begin{align}
\label{eq:L8_rate}
\begin{split}
\Bigg\|
\frac{1}{n}\sum_{t=p}^{n-h}\hat u_t\hat u_t'\hat e_{t,h}^{\,2}
-
\frac{1}{n}\sum_{t=p}^{n-h}u_tu_t'e_{t,h}^{\,2}
\Bigg\|_{\max}
&=
O_p\!\left(
k_A\big\|(\widehat{\mathbf A}^{(\mathrm{re})})^h-\mathbf A^h\big\|_\infty
\right)
+
O_p\!\left(
k_A^2\|\widehat{\mathbf A}^{(\mathrm{re})}-\mathbf A\|_\infty
\right).
\end{split}
\end{align}
Since $h$ is fixed,
\[
\big\|(\widehat{\mathbf A}^{(\mathrm{re})})^h-\mathbf A^h\big\|_\infty
=
O_p\!\left(
k_A^2\|\widehat{\mathbf A}^{(\mathrm{re})}-\mathbf A\|_\infty
\right),
\]
and therefore the leading term in \eqref{eq:L8_rate} is
$O_p\!\left(k_A^3\|\widehat{\mathbf A}^{(\mathrm{re})}-\mathbf A\|_\infty\right)$.
Under the maintained rate
$\|\widehat{\mathbf A}^{(\mathrm{re})}-\mathbf A\|_\infty
=O_p\!\left(k_A^{1.5}(\nu_n/n)^{(1-\mu)/2}\right)$, we obtain
\[
\Bigg\|
\frac{1}{n}\sum_{t=p}^{n-h}\hat u_t\hat u_t'\hat e_{t,h}^{\,2}
-
\frac{1}{n}\sum_{t=p}^{n-h}u_tu_t'e_{t,h}^{\,2}
\Bigg\|_{\max}
=
O_p\!\left(
k_A^{4.5}\big(\nu_n/n\big)^{(1-\mu)/2}
\right)
=
o_p(1).
\]
The same argument applies to the finitely many cross terms $(a,b)\neq(0,0)$, and
thus $n^{-1}\sum_{t=p}^{n-h}(\hat z_t^{\,2}-z_t^2)=o_p(1)$. Combining with Step~1
yields $v'(\Omega_{U,h}^{(\mathrm{hc})}-\Omega_{U,h})v=o_p(1)$.

\medskip
\noindent\textit{HAC.}
The HAC estimator is formed directly from the original score process
$U_t e_{t,h}$. In particular,
\begin{align}
\label{eq:hac_form_direct}
v'\Omega_{U,h}^{(\mathrm{hac})}v
=
\sum_{|\ell|\le b_n} K(\ell/b_n)\,\hat\gamma_\ell,
\qquad
\hat\gamma_\ell
:=
\frac{1}{n}\sum_{t=p+|\ell|}^{n-h}
\bigl(v'\hat U_t\hat e_{t,h}\bigr)\bigl(v'\hat U_{t-\ell}\hat e_{t-\ell,h}\bigr),
\end{align}
and let $\gamma_\ell:=\mathbb E\!\left[(v' U_t e_{t,h})(v' U_{t-\ell} e_{t-\ell,h})\right]$ so that
$v'\Omega_{U,h}v=\sum_{\ell=-\infty}^{\infty}\gamma_\ell$.

Fix any integer $\ell$. Write
\[
\hat\gamma_\ell-\gamma_\ell
=
\Bigg(
\frac{1}{n}\sum_{t=p+|\ell|}^{n-h}
(v' U_t e_{t,h})(v' U_{t-\ell} e_{t-\ell,h})
-\gamma_\ell
\Bigg)
+
\Big(\hat\gamma_\ell-\tilde\gamma_\ell\Big),
\]
where
\[
\tilde\gamma_\ell
:=
\frac{1}{n}\sum_{t=p+|\ell|}^{n-h}
\bigl(v' U_t e_{t,h}\bigr)\bigl(v' U_{t-\ell} e_{t-\ell,h}\bigr).
\]
The first term is a clean-term sample autocovariance error and is $o_p(1)$ for
each fixed $\ell$ by the same weak-dependence and moment conditions as in the HC
part.

For the plug-in error $\hat\gamma_\ell-\tilde\gamma_\ell$, note that
\[
|\hat\gamma_\ell-\tilde\gamma_\ell|
\le
\Bigg\|
\frac{1}{n}\sum_{t=p+|\ell|}^{n-h}
\hat U_t\hat U_{t-\ell}'\hat e_{t,h}\hat e_{t-\ell,h}
-
\frac{1}{n}\sum_{t=p+|\ell|}^{n-h}
U_t U_{t-\ell}' e_{t,h} e_{t-\ell,h}
\Bigg\|_{\max}.
\]
Since $\ell$ is fixed, the same clean-term and estimation-bias argument used to
prove \eqref{represen:block} applies to the lagged cross-products
$u_tu_{t-\ell}'e_{t,h}e_{t-\ell,h}$, yielding
\[
\Bigg\|
\frac{1}{n}\sum_{t=p+|\ell|}^{n-h}
\hat U_t\hat U_{t-\ell}'\hat e_{t,h}\hat e_{t-\ell,h}
-
\frac{1}{n}\sum_{t=p+|\ell|}^{n-h}
U_t U_{t-\ell}' e_{t,h} e_{t-\ell,h}
\Bigg\|_{\max}
=o_p(1),
\]
and hence $\hat\gamma_\ell-\tilde\gamma_\ell=o_p(1)$ (uniformly over fixed
$\ell$) because $\|v\|_1=1$. Therefore $\hat\gamma_\ell-\gamma_\ell=o_p(1)$ for
each fixed $\ell$.

Standard HAC consistency then yields
\[
\sum_{|\ell|\le b_n} K(\ell/b_n)\hat\gamma_\ell
-
\sum_{\ell=-\infty}^{\infty}\gamma_\ell
=o_p(1),
\]
using continuity of $K$ at $0$ with $K(0)=1$, $b_n\to\infty$, $b_n/n\to0$, and
absolute summability of $\{\gamma_\ell\}$. Hence
$v'(\Omega_{U,h}^{(\mathrm{hac})}-\Omega_{U,h})v=o_p(1)$.

\end{proof}

\end{document}